\documentclass[a4paper,12]{article}

\usepackage{format}
\usepackage{graphicx}
\usepackage{epsf}
\usepackage{rsfs}
\usepackage{bbm}
\usepackage{tikz-cd}

\title{Confounding Ghost Channels and Causality: \\
A New Approach to Causal Information Flows} 

\author{Nihat Ay${}^{1,2,3}$}

\begin{document}

\maketitle

\begin{center}
${}^{1}$Max Planck Institute for Mathematics in the Sciences, Leipzig, Germany \\
${}^{2}$Leipzig University, Leipzig, Germany \\
${}^{3}$Santa Fe Institute, Santa Fe, NM, USA \\
\end{center}

\begin{abstract} 
Information theory provides a fundamental framework for the quantification of information flows through channels, formally Markov kernels.  
However, quantities such as mutual information and conditional mutual information do not necessarily reflect the causal nature 
of such flows. We argue that this is often the result of conditioning based on $\sigma$-algebras that are not associated with the given channels.    
We propose a version of the (conditional) mutual information based on families of $\sigma$-algebras that are coupled with the underlying 
channel. This leads to filtrations, which allow us to prove a corresponding causal chain rule as a basic requirement within the presented approach. \\
 

{
\small
\tableofcontents
}
\end{abstract}


\section{Introduction: Information theory and causality}  \label{intro}
Classical information theory \cite{Sh48} is based on the definition of {\em Shannon entropy\/}, a measure of uncertainty about the outcome of a variable  $Z$:
\begin{equation}
      H(Z)  \; = \; - \sum_z p(z) \log p(z), 
\end{equation} 
where $p(z) = {\Bbb P}(Z = z)$ denotes the distribution of $Z$.  
(Throughout this introduction, we consider only variables $X$, $Y$, and $Z$ with finite state sets $\mathsf{X}$, $\mathsf{Y}$, and $\mathsf{Z}$, respectively.)
Shannon entropy serves as a building block of further important quantities. The flow of information from a sender $X$ to a receiver $Z$, for instance, can be 
quantified as the reduction of uncertainty about the outcome of $Z$ based on the outcome of $X$. 
More precisely, we compare two uncertainties here, 
the uncertainty about the outcome of $Z$, 
that is $H(Z)$, 
with the uncertainty about the outcome of $Z$ {\em after\/} knowing the outcome of $X$, that is 
\begin{equation}
      H(Z | X) \; = \;  - \sum_{x} p(x) \sum_{z} p(z | x) \log p(z|x),  
\end{equation}
where $p(z | x) = {\Bbb P}(Z = z | X = x)$ denotes the conditional distribution of $Z$ given $X$. 
Naturally, the latter uncertainty is smaller than or equal to the first one, leading to another fundamental quantity of information theory, the 
{\em mutual information\/}: 
\begin{equation} \label{mutinf1}
    I(X ; Z) \; = \; H(Z) - H(Z  |  X). 
\end{equation}
This difference can also be expressed in geometric terms as the {\em KL-divergence\/} 
\begin{equation} \label{distbetweenchannels}
    I(X ; Z) \; = \; \sum_{x} p(x) \sum_{z} p(z | x) \log \frac{p(z | x)}{p(z)}.
\end{equation}
The KL-divergence plays an important role in information geometry as a canonical divergence \cite{AN00, Am16, AJLS17, AyAm15}.
Such a divergence is characterised in terms of natural geometric properties. It is remarkable that this purely geometric approach
yields the fundamental information-theoretic quantities which were previously derived from a set of axioms that are formulated in non-geometric terms.  
\medskip

Typically, the conditional distribution $p(z | x)$ is interpreted mechanistically as a channel which receives $x$ as an input and generates $z$ as an output. 
In this interpretation, the stochasticity of a channel is considered to be the effect of external or hidden disturbances of a deterministic map. 
This is formalised in terms of a so-called {\em structural equation\/}
\begin{equation} \label{structural}
           Z = f(X , U), 
\end{equation}
with a deterministic map $f$ and a noise variable $U$ that is independent of $X$ \cite{Pe00}. 
With the representation (\ref{structural}),  
the conditional probability distribution can be interpreted as a (probabilitic) causal effect of $X$ on $Z$. 
This interpretation provides the basis for Pearl's influential proposal of a general theory of causality \cite{Pe00}.    
The mutual information (\ref{mutinf1}) then becomes a measure of the causal information flow from $X$ to $Z$ \cite{AP08}, which is consistent with Shannon's 
original idea of the amount of information transmitted through a channel \cite{Sh48}. This consistency, however, is apparently violated when dealing with 
variations or extensions of the sender-receiver setting. We are now going to highlight instances of such inconsistency that will play an important role in this article.  

\section{Confounding ghost channels}
The mutual information is symmetric, that is
$I(X ; Z) = I(Z; X)$. Interpreting it as a measure of causal information flow,   
this symmetry suggests that we have the same amount of causal information flow in both directions, even though the channel goes from 
$X$ to $Z$ so that there cannot be any flow of information in the opposite direction. What is wrong here? This apparent problem, let us call it ``the symmetry puzzle'',   
can be resolved quite easily.  We {\em can\/} revert the direction and compute the conditional distribution $p(x | z) = \frac{p(x)}{p(z)} \, p(z | x)$, 
based on elementary rules of probability theory and
without reference to any mechanisms. 
Furthermore, this conditional distribution 
{\em can\/} be mechanistically interpreted and represented in terms of a structural equation (\ref{structural}). (This is always possible for a given conditional distribution.)  
Such a representation introduces a hypothetical channel
for generating the reverted conditional distribution $p(x | z)$, a kind of ``ghost'' channel that is actually not there.    
The mutual information then quantifies the causal information flow of this hypothetical channel. 
The symmetry of the mutual information simply means that the actual causal information flow in forward direction will be equal to the causal information flow of any hypothetical 
channel in backward direction that is capable of generating the conditional distribution $p(x | z)$. The symmetry puzzle, however, is not the only apparent inconsistency between
the (conditional) mutual information and causality. We are now going to highlight another problem, which is closely related to the symmetry puzzle but requires a deeper analysis for its solution.    
\medskip
 
We now assume that the channel receives $x$ and $y$ as inputs and generates $z$ 
as an output. 
With the corresponding conditional distribution $p(z | x, y) = {\Bbb P}(Z = z | X = x, Y = y)$ we have the {\em conditional mutual information\/} 
of $Y$ and $Z$ given $X$:   
\begin{eqnarray}
    I(Y ; Z | X) 
       & = & H(Z | X) - H(Z  |  X , Y) \label{entropdiff} \\
       & = & \sum_{x , y} p(x, y) \sum_{z} p(z | x , y) \log \frac{p(z | x , y)}{p( z  | x)}. \label{kldiv}
\end{eqnarray}
According to (\ref{entropdiff}), the conditional mutual information compares the uncertainty about $z$ given $x$, 
before and after observing the outcome $y$, 
reflected by the conditional probabilities $p(z|x)$ and $p(z| x, y)$, respectively. 
The representation (\ref{kldiv}) makes this comparison more explicit as a deviation of $p(z | x , y)$ from $p(z | x)$. 
Together with (\ref{distbetweenchannels}), we obtain the chain rule 
\begin{eqnarray} 
I(X, Y ; Z) & = & \sum_{x,y} p(x,y) \sum_{z} p(z | x,y) \log \frac{p(z | x,y)}{p(z)} \\
                & = &  \sum_{x,y} p(x,y) \sum_{z} p(z | x,y) \left[ \log \frac{p(z | x)}{p(z)}  +  \log \frac{p(z | x,y)}{p(z|x)} \right] \label{decompoln} \\
                & = & I(X ; Z) + I(Y; Z |X) . \label{twoterms}
\end{eqnarray}
For the computation of both terms, $I(X ; Z)$ and $I(Y; Z |X)$, we have to evaluate the ``reduced'' conditional distribution $p(z | x)$. 
It is obtained from the original one in the following way: 
\begin{equation} \label{condprob}
    p(z | x) \; = \; \sum_{y} p(y | x) \, p(z | x, y). 
\end{equation}
This conditional distribution represents a second kind of hypothetical channel, a ``ghost channel'', which 
screens off the actual flow of information. It can be sensitive to information about $x$ that is not necessarily 
employed by the original channel $p(z | x,y)$. More precisely, given two states $x$, $x'$ that satisfy 
$p(z | x,y) = p(z| x',y)$ for all $z$ and all $y$ we cannot expect $p(z| x) = p(z | x')$ for all $z$. 
This is a consequence of the coupling through $p(y | x)$, on the RHS of (\ref{condprob}). 
In the most extreme case, $y$ is simply a deterministic map of $x$, so that the knowledge of 
$y$ does not provide any additional information about $z$, that is $p(z| x,y) = p(z | x)$. In the following example we study this  
case more explicitly and thereby highlight the inconsistency of the terms $I(X ; Z)$ and $I(Y ; Z |X)$ in (\ref{twoterms})
with the underlying causal structure. 
We will argue that the conditional distribution (\ref{condprob}) has to be modified in order to allow for a causal interpretation.       
\medskip

\begin{example} \label{copycor}
Consider three variables $X,Y,Z$ with values $-1$ and $+1$, and assume that 
$Z$ is obtained as a copy of $Y$, that is
\begin{equation} \label{copyy}
     p(z | x,y) = 
     \left\{
        \begin{array}{c@{,\quad}l} 
           1 & \mbox{if $z = y$} \\
           0 & \mbox{otherwise} 
        \end{array} 
     \right. .   
\end{equation}  
This means that all information required for the output $Z$ is contained in $Y$. Intuitively, we would expect from a measure 
of information flow to assign zero for the flow from $X$ to $Z$ and a positive value to the flow of information from $Y$ to $Z$ given $X$. 
This is however not what we get with the usual definitions of mutual information and conditional mutual information. The reason for that is the stochastic dependence 
of the inputs $X$ and $Y$. To be more precise, let us assume that the input distribution is given as 
\begin{equation} \label{jointdistr}
          p(x,y) \; = \; \frac{e^{\beta \, x y}}{\sum_{x',y' \in \{\pm 1\}} e^{\beta \, x' y'}} ,
\end{equation}
where the parameter $\beta$ controls the coupling of the inputs. 
This implies $p(x) = {\Bbb P}(X = x) = 1/2$ and $p(y) = {\Bbb P}(Y = y) = 1/2$ for all $x,y \in \{\pm 1\}$. 
We can decompose the full mutual information, as a measure of information flow from $X$ and $Y$ together to $Z$, in the following way 
\begin{equation} \label{decompo1}
  I_\beta(X , Y ; Z) \, = \, \underbrace{I_\beta(Y ; Z)}_{= \, \log 2} + \underbrace{I_\beta(X ; Z | Y)}_{= \, 0} .
\end{equation}
(The subscript $\beta$ indicates the dependence of the respective information-theoretic quantities on this parameter.)
This is consistent with the intuition that $Z$ is receiving all information from $Y$ and no information from $X$. However, we observe an inconsistency if we decompose 
the full mutual information in a different way:   
\begin{equation} \label{decompo2}
   I_\beta (X,Y ; Z) \, = \, I_\beta (X; Z) + I_\beta (Y ; Z | X) .
\end{equation}
For the two terms on the RHS of (\ref{decompo2}) we obtain 
\begin{eqnarray*}
     I_\beta(X;Z) & = & \log (2) - \frac{\log(1 + e^{2 \beta})}{1 + e^{2 \beta}} - \frac{\log(1 + e^{- 2 \beta})}{1 + e^{- 2 \beta}}, \\
     I_\beta(Y ; Z | X) & = &   \frac{\log(1 + e^{2 \beta})}{1 + e^{2 \beta}} + \frac{\log(1 + e^{- 2 \beta})}{1 + e^{- 2 \beta}} .
\end{eqnarray*}
These functions are shown in Figure \ref{beta}. In the limit $\beta \to + \infty$ the two inputs become completely correlated with support $(-1,-1)$ and $(+1,+1)$. 
Correspondingly, for $\beta \to - \infty$ we have complete anti-correlation, and the support is $(-1,+1)$ and $(+1,-1)$.
With (\ref{decompo2}), we obtain the following decomposition: 
\begin{eqnarray}
    I(X , Y ; Z) & = & \lim_{\beta \to \infty}  I_\beta(X , Y ; Z)  \nonumber \\
                     & = & \lim_{\beta \to \infty} I_\beta (X; Z) + \lim_{\beta \to \infty} I_\beta (Y ; Z | X) \nonumber \\
                     & = & \underbrace{I (X; Z)}_{ = \, \log 2} + \underbrace{I (Y ; Z | X)}_{ = \, 0}. \label{decompo2lim}
\end{eqnarray}
The decomposition (\ref{decompo2lim}) gives the impression that $Z$ is receiving all information from 
$X$ and no information from $Y$. However, we know, by construction of this example, that this is not the case. 
The actual situation is better reflected by the decomposition (\ref{decompo1}).  
\begin{figure}[h]
\begin{center}
           \includegraphics[width=8cm]{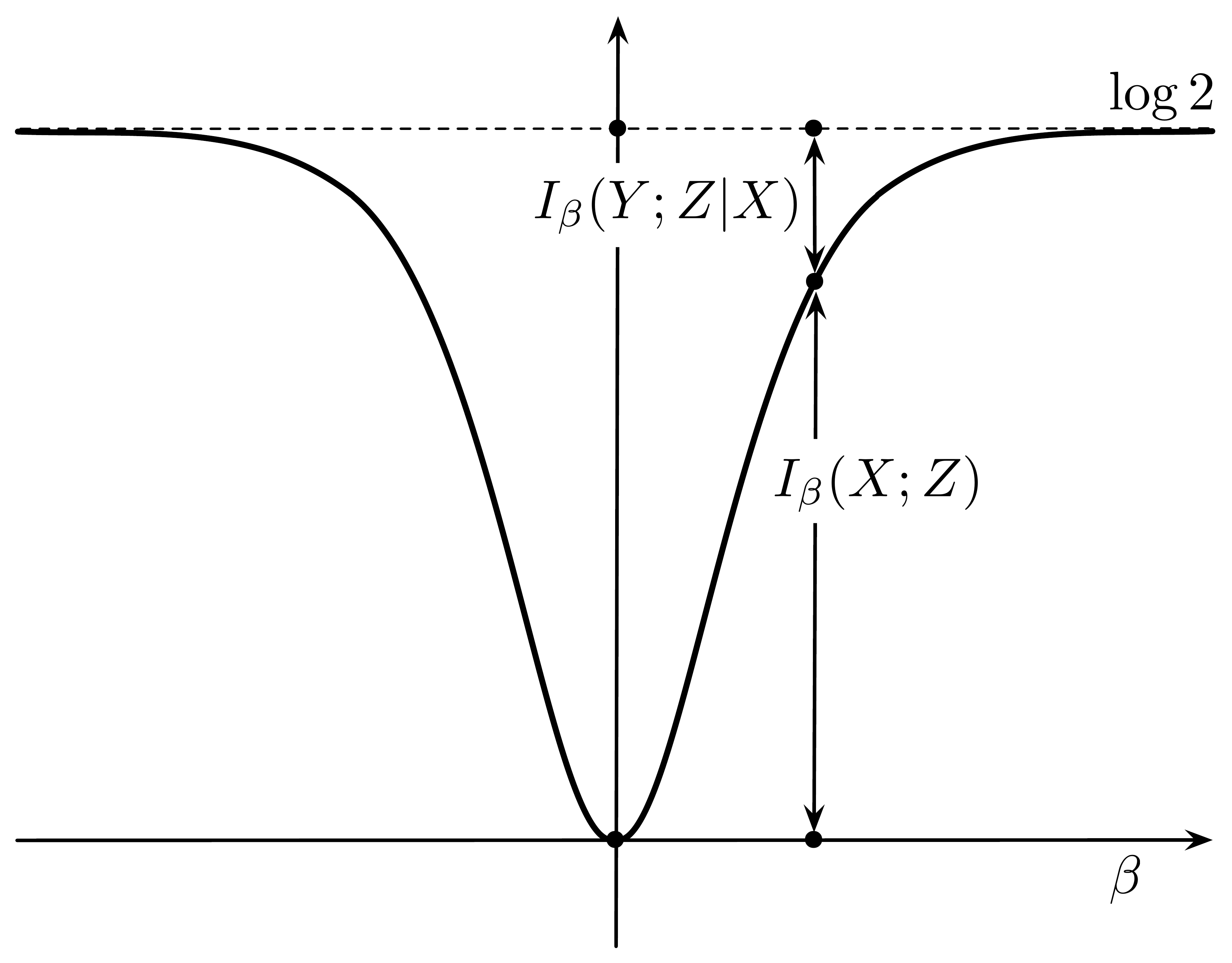}
\caption{The mutual information $I_\beta(X;Z)$ and the conditional mutual information $I_\beta(Y;Z | X)$ as functions of $\beta$. 
Even though the channel does not employ any information from $X$, the mutual information $I_\beta (X; Z)$ converges to the 
maximal value for $\beta \to \infty$.}  
\end{center}
\label{beta}
\end{figure}   
\end{example} 

The problem highlighted in Example \ref{copycor} can be resolved by an appropriate modification of the conditional probability 
(\ref{condprob}). We are now going to outline this modification, which will provide the main idea of this article. 
In a first step, let us assume that $\bar{y}$ is fixed as an input to the channel. Which information about $x$ does the channel then use for generating $z$? In order to qualitatively describe 
that information, we lump any two states $x$ and $x'$ together whenever the channel cannot distinguish them, that is
\[
      p(z | x , \bar{y}) \; = \; p(z | x' , \bar{y})
\]     
for all $z$. This defines a partition $\alpha_{X,\bar{y}}$ of the state set of $X$ that depends on $\bar{y}$. In a second step, we consider the join of all these partitions, 
that is their coarsest refinement. More precisely, we define 
\begin{equation} \label{partitionalpha}
      \alpha_X \; := \; \bigvee_{\bar{y} \in \mathsf{Y}} \alpha_{X,\bar{y}} \; := \; 
      \left\{ \bigcap_{\bar{y} \in \mathsf{Y}} A_{\bar{y}}  \; : \; 
      A_{\bar{y}} \in \alpha_{X, \bar{y}}, \;\; \bar{y} \in \mathsf{Y} \right\}. 
\end{equation}
The partition $\alpha_X$ represents a qualitative description of the information in $X$ that is used by the channel $p(z | x, y)$. 
Denote by $A_x$ the set in $\alpha_X$ that contains $x$. 
When the channel receives $x$, in addition to $y$, then it does not ``see'' the full $x$ but only the class $A_x$, and it is easy to verify $p(z | x,y) = p(z | A_x, y)$. 
Therefore we replace the conditioning $p(z | x)$ in the above formula $(\ref{condprob})$ by 
\begin{equation} \label{condi2}
   \hat{p}(z | x) \; := \; p(z | A_x) 
                       \; = \; \sum_{y} p(y | A_x) \, p(z | A_x , y) 
                       \; = \; \sum_{x' \in A_x} p(x' | A_x) \, p(z | x') .
\end{equation}
This shows that the new conditional distribution $\hat{p}(z | x)$ is obtained by averaging the previous one, $p(z | x)$, 
according to the information that is actually used by the channel $p(z|x,y)$. Now, replacing in (\ref{decompoln}) the conditional distribution 
$p(z |x )$ by $\hat{p}(z | x)$ leads to a corresponding modification of the mutual information and the conditional mutual information: \\
\begin{eqnarray}
   I(X \to Z)       & := &  \sum_{x} p(x) \sum_{z} \hat{p}(z | x) \log\frac{\hat{p}(z | x)}{p(z)} \\
   I(Y \to Z | X) & := & \sum_{x , y} p(x, y) \sum_{z} p(z | x , y) \log \frac{p(z | x , y)}{\hat{p}( z  | x)} .
\end{eqnarray}
It is easy to see that  
\begin{equation}
     I(X \to Z) \; \leq \; I(X ; Z), \qquad  I(Y \to Z | X) \; \geq \; I(Y ; Z | X).
\end{equation}
However, the sum does not change and we have the chain rule 
\begin{equation}
     I(X,Y; Z) \; = \;  I(X \to Z) +  I(Y \to Z | X) .
\end{equation}
With this new definition, we come back to Example \ref{copycor}. The channel defined by (\ref{copyy}) does not use any information from $X$. Therefore, 
$\alpha_{X,\bar{y}} = \{\mathsf{X}\}$ for all $\bar{y} \in \mathsf{Y}$, which implies $\alpha_X = \{ \mathsf{X}\}$. With formula (\ref{condi2}) we obtain 
$\hat{p}(z | x) = p(z | \mathsf{X}) = p(z)$, and therefore 
\begin{equation} \label{decomp2}
     I(X \to Z) \, = \, 0 , \qquad \mbox{and} \qquad I(Y \to Z | X) \, = \, \log 2.
\end{equation}
If we compare this with (\ref{decompo2lim}) we see that the information is shifted from the first to the second term
which corresponds to the variable that has the actual causal effect on $Z$. On the other hand, 
in both cases the sum of the two contributions equals $\log 2$, the full mutual information $I(X,Y ; Z)$. 
\medskip

Causality plays an important role in time series analysis. In this context, {\em Granger causality\/} \cite{Gra69, Gra80} has been the subject of extensive debates which tend 
to highlight its non-causal nature. 
Schreiber proposed an information-theoretic quantification of Granger causality, referred to as {\em transfer entropy\/}, which 
is based on conditional mutual information \cite{Schr00, BBHL16}. Even though transfer entropy is an extremely useful and 
widely applied quantity, it is generally accepted that it has shortcomings as a measure of causal information flow. 
In particular, it can vanish in cases where the causal effect is the strongest possible. We argue that this is again a result of a ghost channel that is involved in the computation of the classical conditional mutual information and screens off the actual causal information flow. This is demonstrated in the following example which is taken from \cite{AP08}. 
Essentially, this example is a reformulation of Example \ref{copycor}, thereby adjusted to the context of time series and stochastic processes.

\begin{example}[Transfer entropy] Consider a  stochastic process $(X_m,Y_m)$, $m = 1,2,\dots$, with state space $\{\pm 1\}^2$
and define $X^m := (X_1,\dots,X_m)$ and $Y^m := (Y_1,\dots, X_m)$. 
The transfer entropy at time $m$ is defined as 
\[
     T(Y^{m-1} \to X_m) \; := \; I(Y^{m-1}; X_m | X^{m - 1}). 
\]
Thus, the transfer entropy quantifies how much information the variables $Y_1, \dots, Y_{m - 1}$ contribute to the evaluation of $X_m$, in addition to the information in $X_1,\dots,X_{m - 1}$.  
We assume that the process is a Markov chain,
given by a transition matrix of the form  
\[
      p(x',y' | x , y) \; = \; p (x' | x,y) \, p(y' | x,y),   
\]
where 
\[
       p (x' | x,y) = \frac{1}{1 + e^{2 \beta x' y}}, \qquad  p (y' | x,y) = \frac{1}{1 + e^{2 \beta y' y}} .
\]
The causal structure of the dynamics is represented by the following diagram: 
\[
\begin{tikzcd}
Y_1 \arrow[r] \arrow[dr] & Y_2 \arrow[r] \arrow[dr] & Y_3    \arrow[r] \arrow[dr]   & \dots \arrow[r] \arrow[dr] & Y_{m-1}  \arrow[r] \arrow[dr] & Y_m \\
X_1                               & X_2                                &  X_3                                   & \dots & X_{m-1}                                & X_m
\end{tikzcd}
\]
As a stationary distribution we have 
\[
     p(+1,+1) \, = \, p(-1, -1) \, = \, \frac{1}{2} - ab, \qquad p(+1,-1) \, = \, p(-1, +1) \, = \, ab, 
\]
where 
\[
              a = \frac{1}{1 + e^{2 \beta}}, \qquad  b = \frac{1}{1 + e^{- 2 \beta}}. 
\]
The transfer entropy can be upper bounded as follows (the subscript $\beta$ indicates the dependence on the coupling parameter $\beta$):  
\begin{eqnarray*}
   T_\beta(Y^{m - 1} \to X_m) 
      & = & I_\beta(Y^{m - 1} ; X_m | X^{m - 1}) \\
      & = & I_\beta(Y_{m - 1} ; X_m | X^{m - 1}) \\
      & = & H_\beta(X_m | X^{m - 1}) - H_\beta(X_m | X^{m - 1}, Y_{m - 1}) \\
      & = & H_\beta (X_m | X^{m - 1}) - H_\beta(X_m | X_{m - 1}, Y_{m - 1}) \\ 
      & \leq & H_\beta(X_m | X_{m - 1}) - H_\beta(X_m | X_{m - 1}, Y_{m - 1}) \\
      & = & I_\beta(Y_{m - 1} ; X_m | X_{m - 1}) .
\end{eqnarray*}
For $\beta = 0$, we have an i.i.d. process with uniform distribution over the states $(+1, + 1)$, $(-1, + 1)$, $(+1, - 1)$, and $(-1,-1)$. For $\beta \to \infty$, we obtain the deterministic transition 
\[
     (x,y) \; \mapsto \; (-y, -y).
\] 
In this limit, the variables $(X_m,Y_m)$ are completely correlated with $p(+1,+1) = p(-1,-1) = \frac{1}{2}$. In both cases, $\beta = 0$ and $\beta \to \infty$, the 
conditional mutual information $I_\beta (Y_{m-1} ; X_m | X_{m - 1})$, and therefore the transfer entropy $T_\beta(Y^{m - 1} \to X_m)$, vanishes. For $\beta = 0$, this does 
not represent a problem because any measure of causal information flow should vanish in the i.i.d. case. However, for $\beta \to \infty$, 
the variable $X_m$ is causally determined by $Y_{m -1}$. Therefore, a measure of casual information flow should be maximal in this case. 
This is not reflected by the transfer entropy. Let us compare this with the information flow measure proposed in this article. 
Given that $X_m$ only depends on $Y_{m - 1}$, the partition (\ref{partitionalpha}) is trivial, that is $\alpha = \{\mathsf{X}\}$. Therefore, 
\[
     I(Y^{m - 1} \to X_m | X^{m - 1}) \; = \; I(Y_{m - 1} ; X_m).
\]   
This quantity is converging to the maximal value $\log 2$ for $\beta \to \infty$. For comparison, both functions are plotted in Figure \ref{transfer}. 
\begin{figure}[h]
\begin{center}
           \includegraphics[width=8cm]{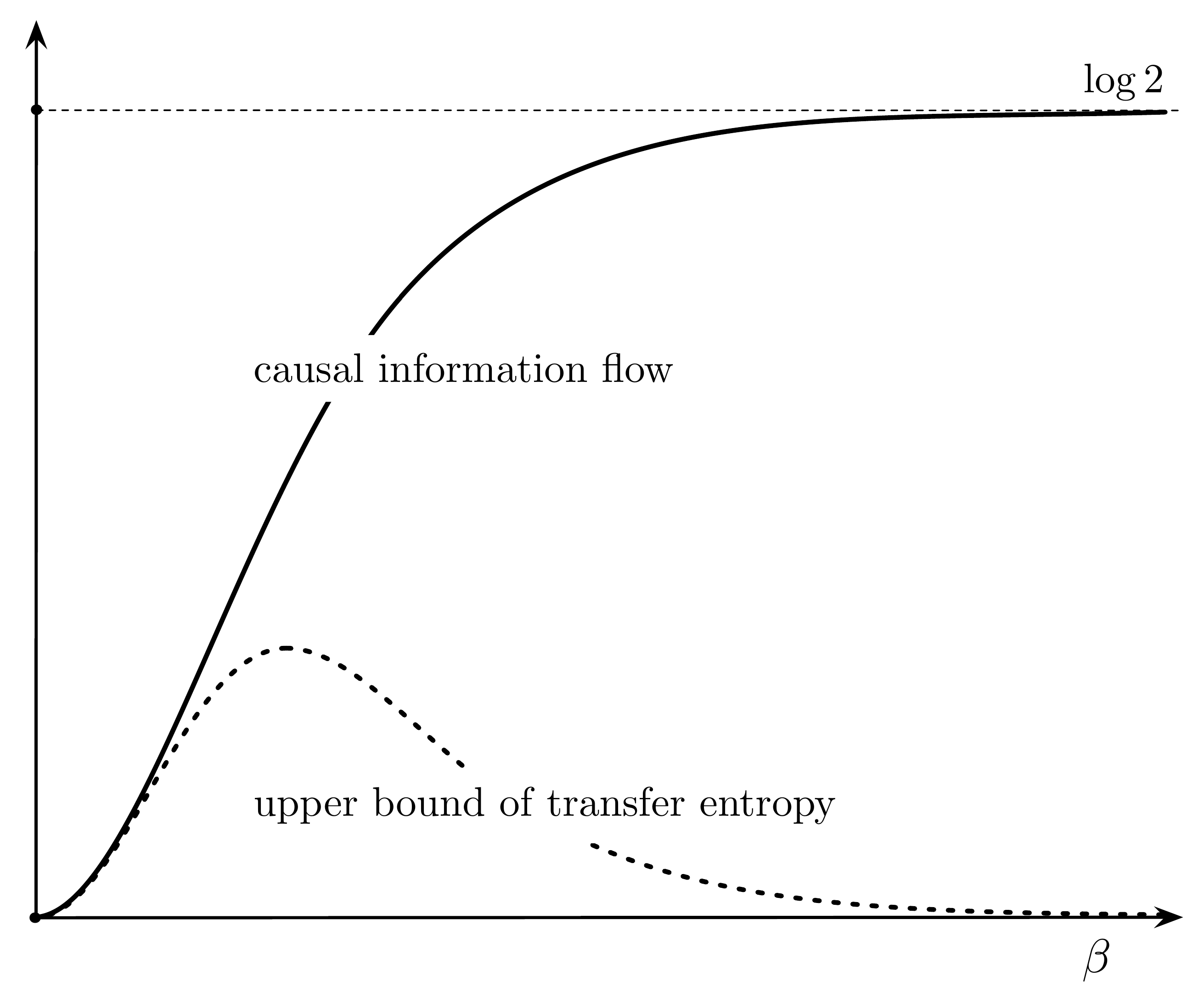}
\caption{Dashed line: 
the conditional mutual information $I_\beta(Y_{m - 1}; X_m | X_{m - 1})$  as an upper bound of the transfer entropy $T_\beta(Y^{m - 1} \to X_m)$; 
solid line: the 
causal information flow $I_\beta(Y^{m - 1} \to X_m | X^{m - 1})$ which coincides with the mutual information 
$I_\beta(Y_{m - 1};  X_m)$ in this example.}  
\end{center}
\label{transfer}
\end{figure}   
\end{example}

In what follows, we will extend the idea outlined in this section to a more general context of measurable spaces, probability measures, and Markov kernels. 
In further steps, 
we will also consider more input nodes. 

\section{General information-theoretic quantities} \label{geninftheo}
In the previous sections, we reviewed fundamental information-theoretic quantities as they are introduced in standard textbooks such as \cite{CT06}.
In this section, we offer an alternative review from a measure-theoretic perspective (see, for instance, \cite{Ka99}). 
This more abstract setting will allow us to identify natural operations and definitions which are not always visible when dealing with finite state spaces.    
   
\paragraph{Shannon entropy}
For a probability space $(\Omega, \rscr{F}, {\Bbb P})$, and a finite measurable 
partition $\gamma = \{C_1, \dots, C_m\}$, that is  $C_i \in \rscr{F}$, $C_i \cap C_j = \emptyset$ for all $i \not= j$, and $\bigcup_{i = 1}^m C_i = \Omega$, 
the {\em Shannon entropy\/} of $\gamma$ is given by
\begin{equation} \label{entropofpart}
    H(\gamma) \; := \; - \sum_{C \in \gamma} {\Bbb P}(C) \log {\Bbb P}(C).
\end{equation}
As a local version of the Shannon entropy, we define
\[
   h(\gamma) \; := \; - \sum_{C \in \gamma} \mathbbm{1}_{C} \log {\Bbb P}(C) ,
\]
where $\mathbbm{1}_{C}$ is the indicator function of $C$. 
Denoting by $C_\omega$ the set in $\gamma$ that contains $\omega \in \Omega$, 
we have $h (\gamma)( \omega ) = - \log {\Bbb P}(C_\omega)$. 
If we integrate the function $h(\gamma) $, we recover the entropy (\ref{entropofpart}) of the partition $\gamma$: 
\begin{eqnarray*}
     \int_{\Omega} h(\gamma) \, {\rm d}{\Bbb P} 
         & = &   - \sum_{C \in \gamma}  \left\{ \int_{\Omega}  \mathbbm{1}_{C}   \, {\rm d} {\Bbb P} \right\} \log {\Bbb P}(C)  \\
         & = & - \sum_{C \in \gamma}  {\Bbb P}(C) \log {\Bbb P}(C)   \\
         & = & H(\gamma).  
\end{eqnarray*}
\paragraph{Conditional entropy}
Consider two finite measurable partitions $\alpha$ and $\gamma$ of $\Omega$, where we assume ${\Bbb P}(A) > 0$ for all $A \in \alpha$. The {\em conditional entropy\/}
of $\gamma$ given $\alpha$ is then defined by  
\begin{equation} \label{condentgen}
  H(\gamma  |  \alpha ) \; := \; - \sum_{A \in \alpha} {\Bbb P}(A) \sum_{C \in \gamma} {\Bbb P}(C  |  A) \log {\Bbb P}(C  |  A) .
\end{equation}
As a local version $h(\gamma | \alpha)$ of $H(\gamma | \alpha)$, we define 
\begin{equation} \label{condentloc}
     h(\gamma | \alpha) \; := \; - \sum_{C \in \gamma} \mathbbm{1}_C \log \left( \sum_{A \in \alpha} {\Bbb P}(C | A) \, \mathbbm{1}_{A}  \right).
\end{equation}
If we evaluate this function for $\omega \in \Omega$ we obtain $h(\gamma | \alpha)( \omega ) = - \log {\Bbb P}(C_\omega | A_\omega)$, 
where $A_\omega$ and $C_\omega$ are the atoms in $\alpha$ and $\gamma$, respectively,  
that contain $\omega$. Integrating $h(\gamma | \alpha)$, we recover (\ref{condentgen}):
\begin{eqnarray*}
\int_{\Omega} h(\gamma | \alpha) \, {\rm d}{\Bbb P} 
   & = & - \int_{\Omega} \log {\Bbb P}(C_\omega | A_\omega) \, {\Bbb P}({\rm d} \omega) \\
   & = &  - \sum_{A \in \alpha} \sum_{C \in \gamma} \int_{A \cap C} \log {\Bbb P}(C_\omega | A_\omega) \, {\Bbb P}({\rm d}\omega) \\
   & = &  - \sum_{A \in \alpha} \sum_{C \in \gamma}  {\Bbb P}(A \cap C) \log {\Bbb P}(C | A) \\
   & = & H(\gamma  |  \alpha ). 
\end{eqnarray*}
The function $h(\gamma | \alpha)$ can be generalised by replacing the partition $\alpha$ by an arbitrary 
$\sigma$-subalgebra $\rscr{A}$ of $\rscr{F}$: 
\begin{equation} \label{extens}
       h(\gamma | \rscr{A}) \; := \; - \sum_{C \in \gamma} \mathbbm{1}_C \log {\Bbb P}(C | \rscr{A}) ,
\end{equation}
where ${\Bbb P}(C | \rscr{A}) = {\Bbb E}(\mathbbm{1}_C | \rscr{A})$. 
Note that this function is only ${\Bbb P}$-almost everywhere defined (abbreviated as ${\Bbb P}$-a.e.). In the case where the $\sigma$-algebra $\rscr{A}$ is given by a finite partition 
$\alpha$ with ${\Bbb P}(A) > 0$ for all $A \in \alpha$, we have  
\[
       {\Bbb P}(C | \rscr{A}) \; = \; \sum_{A \in \alpha} {\Bbb P}(C | A) \, \mathbbm{1}_{A}, \qquad \mbox{${\Bbb P}$-a.e}.   
\]
This shows that the definition (\ref{extens}) is indeed an extension of (\ref{condentloc}). Correspondingly, integrating (\ref{extens}) yields a generalistaion of (\ref{condentgen}):
\begin{eqnarray*}
    H(\gamma | \rscr{A}) & := & \int_{\Omega}  h(\gamma | \rscr{A}) \, {\rm d}{\Bbb P} \\
                & = & - \sum_{C \in \gamma} \int_{\Omega}   \mathbbm{1}_C  \log  {\Bbb P}(C | \rscr{A}) \, {\rm d}{\Bbb P} \\
                & = & - \sum_{C \in \gamma} \int_{\Omega} {\Bbb P}(C | \rscr{A}) \log  {\Bbb P}(C | \rscr{A}) \, {\rm d}{\Bbb P}. 
\end{eqnarray*}

\paragraph{Mutual information}
If we subtract from the entropy of a partition $\gamma$ the conditional entropy of $\gamma$ given a partition $\alpha$, we obtain the {\em mutual information\/}: 
\begin{equation} \label{mutinf}
  I(\alpha  ;  \gamma ) \; := \; - \sum_{A \in \alpha} {\Bbb P}(A) \sum_{C \in \gamma} {\Bbb P}(C  |  A) \log \frac{{\Bbb P}(C  |  A)}{{\Bbb P}(C)} .
\end{equation}
Let us relate this function to the corresponding local functions $h(\gamma)$ and $h(\gamma | \alpha)$. 
Taking the difference, we obtain 
\begin{eqnarray*}
     i(\alpha; \gamma) 
        & := & h(\gamma) - h(\gamma | \alpha) \\
        & = &  - \sum_{C \in \gamma} \mathbbm{1}_{C} \log {\Bbb P}(C) + 
                   \sum_{C \in \gamma} \mathbbm{1}_C \log \left( \sum_{A \in \alpha} {\Bbb P}(C | A) \, \mathbbm{1}_{A}  \right) \\
        & = & \sum_{C \in \gamma} \mathbbm{1}_C \log \left( \sum_{A \in \alpha} \frac{{\Bbb P}(C | A)}{{\Bbb P}(C)} \, \mathbbm{1}_{A} \right)          
\end{eqnarray*}
If we evaluate this for $\omega \in \Omega$ we obtain $i(\alpha; \gamma)(\omega) = \log \frac{{\Bbb P}(C_\omega | A_\omega)}{{\Bbb P}(C_\omega)}$, and thus
we have 
\[
     I(\alpha  ;  \gamma ) \; = \; \int_{\Omega}  i(\alpha; \gamma) \, {\rm d}{\Bbb P}. 
\]
For the general case where the partition $\alpha$ is replaced by a $\sigma$-subalgebra $\rscr{A}$ of $\rscr{F}$, we obtain 
\begin{eqnarray*}
    i(\rscr{A}; \gamma) & := &   h(\gamma) - h(\gamma | \rscr{A}) \\
        & = &  - \sum_{C \in \gamma} \mathbbm{1}_{C} \log {\Bbb P}(C) + 
                   \sum_{C \in \gamma} \mathbbm{1}_C \log {\Bbb P}(C | \rscr{A}) \\ 
        & = & \sum_{C \in \gamma} \mathbbm{1}_C \log \frac{{\Bbb P}(C | \rscr{A})}{{\Bbb P}(C)} .
\end{eqnarray*}
This leads to a corresponding generalisation of (\ref{mutinf}):
\begin{eqnarray} \label{mutinfgen}
     I(\rscr{A}; \gamma) 
        & := & \int_\Omega  i(\rscr{A}; \gamma) \, {\rm d}{\Bbb P} \nonumber \\
        & = & \sum_{C \in \gamma} \int_\Omega   \mathbbm{1}_C \log \frac{{\Bbb P}(C | \rscr{A})}{{\Bbb P}(C)}   \, {\rm d}{\Bbb P} \nonumber \\
        & = & \sum_{C \in \gamma} \int_\Omega   {\Bbb P}(C | \rscr{A})   \log \frac{{\Bbb P}(C | \rscr{A})}{{\Bbb P}(C)}   \, {\rm d}{\Bbb P} .  \label{mutinfgen}
\end{eqnarray}

\paragraph{Conditional mutual information} \label{condmut}
Finally, we define the {\em conditional mutual information\/}. With two $\sigma$-subalgebras $\rscr{A}$ and $\rscr{B}$ of $\rscr{F}$, we define
\begin{equation} \label{locinf}
    i(\rscr{B} ; \gamma | \rscr{A}) \; := \; h(\gamma | \rscr{A}) - h(\gamma | \rscr{A} \vee \rscr{B} ) \; = \;   
               \sum_{C \in \gamma} \mathbbm{1}_C \log \frac{{\Bbb P}(C | \rscr{A} \vee \rscr{B})}{{\Bbb P}(B | \rscr{A})} .
\end{equation}
Integration of this function leads to 
\begin{eqnarray} 
   I(\rscr{B}; \gamma | \rscr{A}) 
      & = & \int_\Omega  i(\rscr{B} ; \gamma | \rscr{A}) \, {\rm d}{\Bbb P} \nonumber \\
      & = &  \sum_{C \in \gamma}  \int_\Omega  \mathbbm{1}_C \log \frac{{\Bbb P}(C | \rscr{A} \vee \rscr{B})}{{\Bbb P}(C | \rscr{A})} \, {\rm d}{\Bbb P} \nonumber \\
      & = & \sum_{C \in \gamma}  \int_\Omega   {\Bbb P}(C | \rscr{A} \vee \rscr{B}) \log \frac{{\Bbb P}(C | \rscr{A} \vee \rscr{B})}{{\Bbb P}(C | \rscr{A})} \, {\rm d}{\Bbb P}.
                \label{condmutinfgen}
\end{eqnarray}
In a final step, we could further extend $i(\rscr{B} ; \gamma | \rscr{A})$ and $I(\rscr{B} ; \gamma | \rscr{A})$ to the case where $\gamma$ is replaced by a 
$\sigma$-algebra $\rscr{C}$, by taking the supremum over all finite partitions $\gamma$ in $\rscr{C}$. However, in this article 
we restrict attention to a fixed finite partitions $\gamma$. 

\section{The chain rule as a guiding scheme} 
\subsection{Two inputs} \label{twoinp}
In the introduction, Section \ref{intro}, 
we have used the two-input case for discrete random variables in order to highlight the main issue with the classical definitions of the mutual information and the 
conditional mutual information and to outline the core idea of this article. 
After having introduced the required information-theoretic quantities for more general variables in Section \ref{geninftheo}, 
we now revisit the instructive two-input case and demonstrate how measure-theoretic concepts come into play here very naturally.    
\medskip

Consider measurable spaces $(\mathsf{X},\rscr{X})$, $(\mathsf{Y},\rscr{Y})$, $(\mathsf{Z},\rscr{Z})$, and their product 
\[
  (\Omega, \rscr{F}) \; := \; (\mathsf{X} \times \mathsf Y \times \mathsf{Z}, \rscr{X} \otimes \rscr{Y} \otimes \rscr{Z}). 
\]  
In order to ensure the existence of various (regular versions of) conditional distributions, we need to assume that these measurable spaces carry 
a further structure. Typically, it is sufficient to require that $(\mathsf{X},\rscr{X})$, $(\mathsf{Y},\rscr{Y})$, and $(\mathsf{Z},\rscr{Z})$ 
are Polish spaces (see \cite{Du02}, Theorem 13.1.1), which will be implicitly assumed hereinafter for all measurable spaces.   \\

Now, consider  
a probability measure $\mu$ on $(\mathsf{X} \times \mathsf{Y}, \rscr{X} \otimes \rscr{Y})$
and a Markov kernel 
\[
   \nu: \mathsf{X} \times \mathsf{Y} \times \rscr{Z} \to [0,1],
\]
which models a channel that takes two inputs, $x \in \mathsf{X}$ and $y \in \mathsf{Y}$, and generates a possibly random output $z \in \mathsf{Z}$. 
This allows us to define a probability measure on the joint space $(\Omega, \rscr{F})$, given by 
\[
     {\Bbb P}(A \times B \times C) \, := \, \int_{A \times B} \nu(x,y; C) \, \mu({\rm d}x,{\rm d}y).
\] 
With the natural projections 
\[
    X : \Omega \to \mathsf{X}, \;\; (x,y,z) \mapsto x, \qquad Y : \Omega \to \mathsf{Y}, \;\; (x,y,z) \mapsto y, \qquad 
    Z : \Omega \to \mathsf{Z}, \;\; (x,y,z) \mapsto z ,
\] 
we have 
\begin{eqnarray}
     \mu(A \times B) & = & {\Bbb P}(X \in A, Y \in B), \label{inpvert} \\
     \nu(x,y; C)         & = & {\Bbb P}(Z \in C | X = x, Y =  y). \label{bedvert}
\end{eqnarray}
Furthermore, we have the marginals 
\begin{eqnarray}
   \mu_X(A) & := & \mu(A \times \mathsf{Y}) \; = \; {\Bbb P}( X \in A ), \\
   \mu_Y (B) & := & \mu(\mathsf{X} \times B) \; = \; {\Bbb P}( Y  \in B ), 
\end{eqnarray}
and, finally, the $\nu$-push-forward measure of $\mu$, 
\begin{equation}
     \mu_\ast (C) \, := \, {\Bbb P}(Z \in C).
\end{equation}

Note that the definition of the conditional distribution ${\Bbb P}(Z \in C | X = x, Y =  y)$ on the RHS of (\ref{bedvert}) is quite general and does not exclude cases where 
${\Bbb P}(X = x, Y= y) = 0$. It requires a formalism that goes beyond the context of variables with finitely many state sets 
$\mathsf{X}$, $\mathsf{Y}$, and $\mathsf{Z}$. It is important to outline this formalism in some detail here. It will provide the basis   
for an appropriate definition of marginal channels. The definition of the conditional distribution 
\begin{equation} \label{conddistr}
     {\Bbb P}(Z \in C | X = x, Y =  y)
\end{equation}
involves two steps: 
\begin{enumerate}
\item[Step 1] We interpret the indicator function $\mathbbm{1}_{\{Z \in C\}}$ as an element of the Hilbert space 
$L^2(\Omega, \rscr{F}, {\Bbb P})$ and project it onto the (closed) linear subspace of $(X,Y)$-measurable functions $\Omega \to {\Bbb R}$.
Its projection is referred to as {\em conditional expectation\/} and denoted by  
\begin{equation} \label{condexpectxy}
           {\Bbb E}( \mathbbm{1}_{\{Z \in C\}} |  X,Y ).
\end{equation}    
Note that the elements of the Hilbert space $L^2(\Omega, \rscr{F}, {\Bbb P})$ are equivalence classes of functions where two functions are identified if they  
coincide on a measurable set of probability one. 
Therefore, the conditional expectation (\ref{condexpectxy}) is almost surely well defined. 
\item[Step 2] Formally, ${\Bbb E}( \mathbbm{1}_{\{Z \in C\}} |  X,Y )$ is a real-valued function defined on $\Omega$. On the other hand, it is $(X,Y)$-measurable so that we should be 
able to interpret it as a function of $x$ and $y$. Indeed, it follows from the factorisation lemma 
that there is a unique measurable function $\varphi_C: (\mathsf{X} \times \mathsf{Y}, \rscr{X} \otimes \rscr{Y}) \to {\Bbb R}$ 
satisfying ${\Bbb E}( \mathbbm{1}_{\{Z \in C\}} |  X , Y ) = \varphi_C \circ (X, Y)$.  
The conditional distribution (\ref{conddistr}) is then simply defined to be the function $\varphi_C$, which has $x$ and $y$ as arguments. 
In the special situation where we start with a Markov kernel $\nu$, we recover it in terms of equation (\ref{bedvert}). It turns out that this equation already describes 
a quite general situation. Under mild conditions, assuming, for instance, that all measurable spaces are Polish spaces, 
the conditional distribution (\ref{conddistr}) can be considered to be a Markov kernel, as a function of $x$, $y$, {\em and\/} $C$. 
\end{enumerate}

For the definition of mutual information and conditional mutual information, as generalisations of (\ref{distbetweenchannels}) and (\ref{kldiv}), respectively, we have 
to find an appropriate notion of a marginal kernel. We begin with the conditional distribution ${\Bbb P}(Z \in C | X = x)$, as generalisation of $p(z | x)$. For its evaluation 
we repeat the arguments of the above two steps and consider the conditional expectation  
\begin{equation} \label{condexpect}
           {\Bbb E}( \mathbbm{1}_{\{Z \in C\}} |  X  ). 
\end{equation}      
This is an $X$-measurable random variable $\Omega \to {\Bbb R}$. 
By the factorisation lemma we have  
a unique measurable function $\nu_X( \cdot ;C) : (\mathsf{X}, \rscr{X}) \to {\Bbb R}$ 
satisfying ${\Bbb E}( \mathbbm{1}_{\{Z \in C\}} |  X  ) =  \nu_X(X ; C)$, and we set
\[
    {\Bbb P}(Z \in C | X = x) \; := \; \nu_X(x;C).
\] 
Under mild conditions we can assume that $\nu_X(x;C)$ defines a Markov kernel when considered as a function $\nu_X: \mathsf{X} \times \rscr{Z} \to [0,1]$ 
in $x$ {\em and\/} $C$.  
\medskip
 
We can now easily extend the classical definitions of mutual information and conditional mutual information to the 
context of this section. For a finite measurable partition $\gamma$ of $\mathsf{Z}$ we can use (\ref{mutinfgen}) to define the mutual informations 
\begin{eqnarray}
   I_\gamma(X, Y; Z)  
      & := & I( \sigma(X) \vee \sigma(Y) ; \gamma)  \nonumber \\
      & = & \sum_{C \in \gamma} \int_{\mathsf{X}\times Y} \nu(x , y ; C) \log \frac{\nu(x, y ; C)}{\mu_\ast(C)} \, \mu({\rm d} x, {\rm d} y) , 
\end{eqnarray}
and 
\begin{eqnarray}   
   I_\gamma(X; Z)    & := & I(\sigma(X) ; \gamma) \nonumber \\
               & = & \sum_{C \in \gamma} \int_{\mathsf{X}} \nu_X(x ; C) \log \frac{\nu_X(x ; C)}{\mu_\ast(C)} \, \mu_X({\rm d} x) . \label{causalv}
\end{eqnarray}
Furthermore, with (\ref{condmutinfgen}) we define the conditional mutual information   
\begin{eqnarray} \label{condmutinfxyz}
   I_\gamma(Y; Z | X) 
      & := & I(\sigma(Y) ; \gamma | \sigma (X)) \nonumber \\
      & = & \sum_{C \in \gamma} \int_{\mathsf{X} \times \mathsf{Y}} \nu(x, y;C) \log \frac{\nu(x,y; C)}{\nu_X(x ; C)} \, \mu({\rm d} x , {\rm d} y) . 
\end{eqnarray}
We repeat the computation (\ref{decompoln}) and 
decompose the mutual information $I_\gamma(X, Y; Z)$ as follows:
\begin{eqnarray} \label{chainr}
    I_\gamma(X, Y; Z) 
       & = &  \sum_{C \in \gamma} \int_{\mathsf{X}\times Y} \nu(x , y ; C) \log \frac{\nu(x, y ; C)}{\mu_\ast(C)} \, \mu({\rm d} x, {\rm d} y) \\  
       & = &  \sum_{C \in \gamma} \int_{\mathsf{X}\times Y} \nu(x , y ; C) \left[  \log \frac{ \nu_X(x ; C) }{\mu_\ast(C)} + 
                  \log \frac{\nu(x, y ; C)}{ \nu_X(x ; C) } \right] \, \mu({\rm d} x, {\rm d} y) \label{divisi}\\  
       & = & I_\gamma(X; Z)  +  I_\gamma(Y; Z | X).  \label{chainr2}
\end{eqnarray}
We argue that, in order to have a causal decomposition 
of the full mutual information $I_\gamma(X, Y; Z)$ into two terms similar to $I_\gamma(X; Z)$ and $I_\gamma(Y; Z | X)$,  
we have to modify the marginal channel 
\begin{equation} \label{condi}
\nu_X(x ; C) \, = \, {\Bbb P}(Z \in C | X = x) 
\end{equation}
in (\ref{divisi}). In this modification, the conditioning with respect to $X$ has to be adjusted to the actual information used by the kernel $\nu(x,y ; C)$. To this end, 
we consider the smallest $\sigma$-subalgebra $\rscr{A}_{X, \bar{y}}$ of ${\rscr X}$ 
for which all constrained Markov kernels $\nu_{X,\bar{y}} (\cdot ; C) := \nu(\cdot, \bar{y} ; C)$, $C \in \rscr{Z}$, 
are measurable. It corresponds to the partition $\alpha_{X,\bar{y}}$
that appears in (\ref{partitionalpha}). 
Now we generalise the definition (\ref{partitionalpha}) of the partition $\alpha_X$ 
by combining all the $\sigma$-algebras $\rscr{A}_{X , \bar{y}}$: 
\begin{equation}
      \rscr{A}_X \, := \,  \bigvee_{\bar{y} \in \mathsf{Y}} \rscr{A}_{X , \bar{y}} \; \subseteq \; \rscr{X}. 
\end{equation}
Note that this $\sigma$-algebra is typically not contained in the $\sigma$-algebra generated by the channel $\nu$, that is the smallest $\sigma$-subalgebra in 
$\rscr{X} \otimes \rscr{Y}$ for which $(x,y) \mapsto \nu(x,y ; C)$ is measurable for all $C \in \rscr{Z}$. This is illustrated by the following example.  

\begin{example} \label{notsubalgebra}
We consider 
\[
     (\mathsf{X}, \rscr{X}) \; = \; (\mathsf{Y}, \rscr{Y}) \; = \; (\mathsf{Z}, \rscr{Z}) \; = \; ({\Bbb R}, \rscr{B}({\Bbb R})),
\]
where $\rscr{B}({\Bbb R})$ denotes the Borel $\sigma$-algebra of ${\Bbb R}$. We assume that the channel $\nu$ is simply given by the addition $(x,y) \mapsto x + y$:
\[
    \nu(x,y; C) \; = \; \mathbbm{1}_{C}(x + y). 
\] 
As $\rscr{B}({\Bbb R})$ is generated by the intervals $[r - \varepsilon, r + \varepsilon] \subseteq {\Bbb R}$, the   
smallest $\sigma$-algebra $\rscr{A} \subseteq \rscr{X} \otimes \rscr{Y}$ for which all functions $\nu(\cdot,\cdot ; C)$ are measurable is generated 
by the following sets 
\[
     A ({r, \varepsilon}) \; := \;
                   \left\{(x,y) \in {\Bbb R}^2 \; : \;  r - \varepsilon \leq x + y \leq r + \varepsilon \right\}, \qquad r \in {\Bbb R}, \; \varepsilon \in {\Bbb R}_{+}. 
\]
Now let us consider $\rscr{A}_{X,\bar{y}}$, the $\sigma$-algebra generated by the kernel 
\[
    \nu_{X,\bar{y}}: \; {\Bbb R} \times \rscr{B}({\Bbb R}) \; \to \; [0,1],  \qquad (x,C) \; \mapsto \; \nu_{X, \bar{y}}(x;C) := \nu(x,{\bar{y}} ; C).  
\]
It is easy to see that $\rscr{A}_{X , \bar{y}}$ is the smallest $\sigma$-subalgebra of $\rscr{X}$ containing the $\bar{y}$-sections
\[
    A_{X, \bar{y}} (r, \varepsilon) \; := \; \left\{ x \in {\Bbb R} \; : \; (x, \bar{y}) \in A (r, \varepsilon ) \right\} \; = \; 
                                            \left\{x \in {\Bbb R} \; : \;  r - \bar{y} - \varepsilon \leq x \leq r  - \bar{y} + \varepsilon \right\}
\]
This example shows that the cylinder sets $A \times {\Bbb R}$, $A \in \rscr{A}_{X , \bar{y}}$, are not necessarily contained in $\rscr{A}$.  
\begin{figure}[h]
\begin{center}
           \includegraphics[width=12cm]{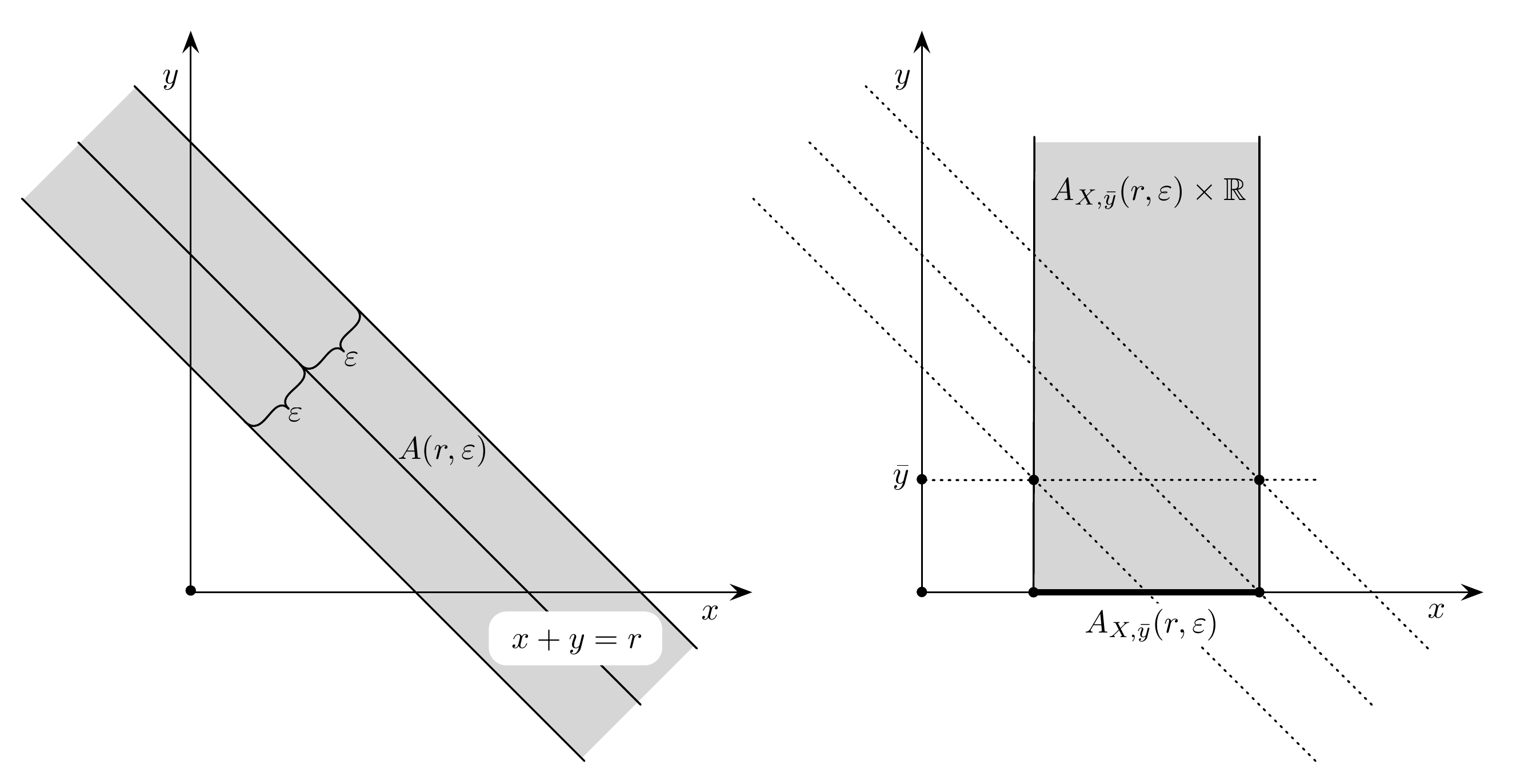}
\caption{Illustration of the $\nu$-measurable sets in ${\Bbb R}^2$ and their sections in ${\Bbb R}$.}  
\end{center}
\end{figure}   
\end{example}

With the $\sigma$-subalgebra $\rscr{A}_X$ of $\rscr{X}$, we can now modify the random variable $X: (\Omega, \rscr{F}, {\Bbb P}) \to (\mathsf{X}, \rscr{X})$ 
by simply reducing the image $\sigma$-algebra to $\rscr{A}_X$:
\[
    \widehat{X} : \; (\Omega, \rscr{F}, {\Bbb P}) \; \to \; (\mathsf{X}, \rscr{A}_X).  
\]
We will see that this step is crucial here, even though it might appear like a minor technical step at first sight.  
It allows us to modify (\ref{condexpect}) by replacing the full $\sigma$-algebra of $X$, $\rscr{X}$, by the 
$\sigma$-algebra of $\widehat{X}$:
\begin{equation} \label{condexpect2} 
      {\Bbb E}( \mathbbm{1}_{\{Z \in C\}} | \widehat{X}).
\end{equation}
This is, by definition, an $\widehat{X}$-measurable random 
variable $\Omega \to {\Bbb R}$. By the factorisation lemma, we can find a unique measurable function 
$\hat{\nu}_X(\cdot ; C) : (\mathsf{X}, \rscr{A}_X) \to {\Bbb R}$ satisfying 
$ {\Bbb E}( \mathbbm{1}_{\{Z \in C\}} | \widehat{X} )=  \hat{\nu}_X( \widehat{X} ; C)$. This yields a new marginal channel,
\[
   {\Bbb P}(Z \in C | \widehat{X} = x) \; := \; \hat{\nu}_X(x ; C)  ,
\] 
as a modification of $\nu_X(x ; C)$ which appears twice in (\ref{divisi}). 
Note that the kernel $\hat{\nu}_X(x;C)$ is defined 
almost surely. However, the definition of a conditional mutual information will be independent of the version of that kernel. 
\medskip

Now we come to the definition of a causal version of the mutual information (\ref{causalv}) and the conditional mutual information (\ref{condmutinfxyz}). 
We simply replace in these definitions $\nu_X(x ; C)$ by $\hat{\nu}_X(x ; C)$:
\begin{eqnarray} 
   I_\gamma(X \to Z)
   & := & I(\sigma(\widehat{X}) ; \gamma) \nonumber \\
   & = & \sum_{C \in \gamma} \int_{\mathsf{X}} \hat{\nu}_X(x ; C) \log \frac{\hat{\nu}_X(x ; C)}{\mu_\ast(C)} \, \mu_X({\rm d} x) , \label{causversmuin} 
\end{eqnarray}
\begin{eqnarray}   
   I_\gamma(Y \to Z | X) 
   & := & I(\sigma(Y) ; \gamma | \sigma(\widehat{X}) ) \nonumber   \\
   & = & \sum_{C \in \gamma} \int_{\mathsf{X} \times \mathsf{Y}} \nu(x, y;C) \log \frac{\nu(x,y; C)}{\hat{\nu}_X(x ; C)} \, \mu({\rm d} x , {\rm d} y) . \label{condmutinfxyz2} 
\end{eqnarray}
The following proposition relates the causal quantities (\ref{causversmuin}) and (\ref{condmutinfxyz2})  
to the corresponding non-causal ones, (\ref{causalv}) and (\ref{condmutinfxyz}). 

\begin{proposition} \label{firstchain}
We have the chain rule 
\begin{equation} \label{chainnew}
 I_\gamma(X,Y;Z) \; = \; I_\gamma(X \to Z) + I_\gamma(Y \to Z | X). 
\end{equation}
Furthermore,
\begin{equation} \label{compa}
      I_\gamma(X \to Z) \; \leq \; I_\gamma(X ; Z), \qquad I_\gamma(Y \to Z | X) \; \geq \; I_\gamma(Y ; Z | X).
\end{equation}
\end{proposition}
\begin{proof} With $C_z$ denoting the set in $\gamma$ that contains $z$,
we have 
\[
   \log \frac{\nu(x,y ; C_z)}{\mu_\ast(C_z)} \; = \; 
   \log \frac{\nu(x,y ; C_z)}{\hat{\nu}_X(x ; C_z)} + \log \frac{\hat{\nu}_X(x ; C_z)}{\mu_\ast(C_z)}.
\]
Integrating this with respect to $\nu(x,y; {\rm d} z)$ we get
\begin{eqnarray}
 \lefteqn{\int_{\mathsf{Z}}  \log \frac{\nu(x,y ; C_z)}{\mu_\ast(C_z)} \, \nu(x,y; {\rm d} z) } \nonumber \\
       & = &  \int_{\mathsf{Z}}  \log \frac{\nu(x,y ; C_z)}{\hat{\nu}_X(x ; C_z)}  \, \nu(x,y; {\rm d} z)  + 
                  \int_{\mathsf{Z}}  \log \frac{\hat{\nu}_X(x ; C_z)}{\mu_\ast(C_z)} \, \nu(x,y; {\rm d} z) \nonumber \\
       & = &  \sum_{C \in \gamma} \nu(x,y ; C)  \log \frac{ \nu(x,y ; C)}{\hat{\nu}_X (x ; C)}  + 
                  \sum_{C \in \gamma} \nu(x,y ; C)  \log \frac{\hat{\nu}_X(x ; C)}{\mu_\ast(C)} .  \label{decomp}
\end{eqnarray}
Further integrating the first term of (\ref{decomp}) with respect to $\mu$ gives us $I_\gamma(Y \to Z | X)$ (see (\ref{condmutinfxyz2})). 
For the corresponding integration of the second term, we obtain
\begin{eqnarray} \label{ineg}
       \lefteqn{\sum_{C \in \gamma} \int_{\mathsf{X} \times \mathsf{Y}} \nu(x,y ; C)  \log \frac{\hat{\nu}_X(x ; C)}{\mu_\ast(C)} \, \mu({\rm d} x , {\rm d} y)} \qquad \qquad \nonumber \\
       & = & \sum_{C \in \gamma} \int_{\Omega} {\Bbb P}(Z \in C | X, Y)  \log \frac{{\Bbb P}(Z \in C | \widehat{X})}{{\Bbb P}(Z \in C)} \, {\rm d} {\Bbb P} \label{step1} \\
       & = & \sum_{C \in \gamma} \int_{\Omega} {\Bbb P}(Z \in C | \widehat{X})  \log \frac{{\Bbb P}(Z \in C | \widehat{X})}{{\Bbb P}(Z \in C)} \, {\rm d} {\Bbb P}  \label{gener} \\
       & = & \sum_{C \in \gamma} \int_{\mathsf{X}} \hat{\nu}_X(x ; C)  \log \frac{\hat{\nu}_X(x ; C)}{\mu_\ast(C)} \, \mu_X({\rm d} x) .   \label{step2}    
\end{eqnarray}
The crucial step (\ref{gener}) follows from the general property of the conditional expectation of a function $f$ with respect to a $\sigma$-subalgebra ${\rscr A}$:
\[
       \int_\Omega f g \, {\rm d}{\Bbb P} = \int_\Omega {\Bbb E}(f | {\rscr A}) g \, {\rm d}{\Bbb P}, \qquad \mbox{for all ${\rscr A}$-measurable functions $g$}.
\]       
Here, $f$ is given by ${\Bbb P}(Z \in C | X, Y)$, $\rscr{A}$ is the $\sigma$-algebra generated by $\widehat{X}$, and $g$ is given by 
$ \log \frac{{\Bbb P}(Z \in C | \widehat{X})}{{\Bbb P}(Z \in C)}$ which is $\widehat{X}$-measurable. The steps (\ref{step1}) and (\ref{step2}) follow directly from the definitions of the
Markov kernels, and we finally obtain $I_\gamma(X \to Z)$ (see (\ref{causversmuin})). This concludes the proof of the chain rule (\ref{chainnew}). 

We now prove the   
inequalities (\ref{compa}) where we can restrict attention to the first one. 
We consider the convex function $\phi(r) := r \log \frac{r}{\mu_\ast (C)}$, for $r > 0$, and $\phi(0) := 0$, and apply Jensen's inequality for conditional expectations: 
\begin{equation} \label{Jensen}
   {\Bbb E} \left( \phi\big( {\Bbb P}(Z \in C | X ) \big) \Big|    \widehat{X} \right) \; \geq \; \phi \left( {\Bbb E} \left( {\Bbb P}(Z \in C | X )  \Big|    \widehat{X} \right) \right)
   \; = \; \phi \left( {\Bbb P}(Z \in C | \widehat{X} )  \right)
\end{equation}
This implies 
\begin{eqnarray*}
   I_\gamma (X ; Z) 
   & = & \sum_{C \in \gamma} \int_{\Omega} {\Bbb P}(Z \in C | {X})  \log \frac{{\Bbb P}(Z \in C | {X})}{{\Bbb P}(Z \in C)} \, {\rm d} {\Bbb P} \\
   & = &  \sum_{C \in \gamma} {\Bbb E} \left( \phi\big( {\Bbb P}(Z \in C | {X})  \big) \right)  \\
   & = &  \sum_{C \in \gamma} {\Bbb E} \left( {\Bbb E} \left( \phi\big( {\Bbb P}(Z \in C | {X})  \big)  \Big| \widehat{X}      \right)     \right)  \\
   & \geq &  \sum_{C \in \gamma}  {\Bbb E} \left(  \phi \left( {\Bbb P}(Z \in C | \widehat{X} )  \right) \right) \qquad \mbox{(by (\ref{Jensen}))}  \\
   & = & \sum_{C \in \gamma} \int_{\Omega} {\Bbb P}(Z \in C | \widehat{X})  \log \frac{{\Bbb P}(Z \in C | \widehat{X})}{{\Bbb P}(Z \in C)} \, {\rm d} {\Bbb P} \\
   & = & I_\gamma( X \to Z) .
\end{eqnarray*}
The second inequality in (\ref{compa}) follows from the first one and the chain rule (\ref{chainnew}). 
\end{proof}

Let us interpret this result. The first inequality in (\ref{compa}) highlights the fact that the stochastic dependence between $X$ and $Z$, here quantified by the 
usual mutual information $I_\gamma(X ; Z)$, cannot be fully attributed to the causal effect of $X$ on $Z$. Some part of $I_\gamma(X ; Z)$ is purely associational, and  $I_\gamma(X \to Z)$ constitutes the causal part of it. 
The second inequality in (\ref{compa}) highlights a different fact. Conditioning on the variable $X$ may ``screen off'' 
some part of the causal effect of $Y$ on $Z$. More precisely, the uncertainty reduction about the outcome of $Z$ through $X$ can be so strong that a further reduction through $Y$ becomes ``invisible''. Therefore, the classical conditional mutual information, $I_\gamma(Y; Z | X)$, tends to reflect only part of the 
causal influence of $Y$ on $Z$ given $X$, $I_\gamma(Y \to Z | X)$. Even though the classical information-theoretic quantities are 
replaced by their causal versions, the full mutual information can still be decomposed 
according to the chain rule (\ref{chainnew}). However, in comparison to the decomposition (\ref{chainr2}), some amount of it is shifted from one term to the other  
so that both terms can be interpreted causally. 
\medskip

It turns out, that the definitions (\ref{causversmuin}) and (\ref{condmutinfxyz2}) require a careful extension if we want to have a general chain 
rule for more than two input variables. 
We are now going to highlight this for three input variables.

\subsection{Three inputs} \label{threeinp}
We now consider three input variables. This will reveal that the previous case with two input variables is quite special. An extension to more 
than two variables requires an adjustment of our definition of causal information flow.   
\medskip

We consider a third input variable (denoted below by $W$) with values in a measurable space $(\mathsf{W}, \rscr{W})$, a probability measure 
\[
   \mu \qquad \mbox{on} \qquad 
  (\mathsf{W} \times \mathsf{X} \times \mathsf{Y}, \rscr{W} \otimes \rscr{X} \otimes \rscr{Y}), 
\]  
and an input-output channel, given by a Markov kernel 
\[
   \nu: \; \mathsf{W} \times \mathsf{X} \times \mathsf{Y} \times \rscr{Z} \; \to \; [0,1]. 
\]
This gives rise to a probability space, consisting of the measurable space 
\[
    (\Omega , \rscr{F}) \,  := \, (\mathsf{W} \times \mathsf{X} \times \mathsf{Y} \times \mathsf{Z} , \rscr{W} \otimes \rscr{X} \otimes \rscr{Y} \otimes \rscr{Z}),
\]
and the probability measure ${\Bbb P}$ defined by
\[
     {\Bbb P}(A \times B \times C \times D) \, := \, \int_{A \times B \times C} \nu(w,x,y ; D) \, \mu({\rm d}w, {\rm d}x, {\rm d}y). 
\] 
Finally, we have the natural projections $W: \Omega \to \mathsf{W}$, $X: \Omega \to \mathsf{X}$, $Y: \Omega \to \mathsf{Y}$, and $Z: \Omega \to \mathsf{Z}$. 
\medskip

The definition of the marginal kernel $\hat{\nu}_X(x ; C)$, as introduced in Section \ref{twoinp}, is directly applicable to the situation of three input variables. 
It allows us to define marginal channels by an appropriate grouping of two input variables into one input variable, which formally reduces 
the three-input case to a two-input case. In particular, we can define the channels $\hat{\nu}_{W,X}(w, x ; C)$ and $\hat{\nu}_W(w ; C)$, 
by grouping $W,X$ and $X,Y$, respectively, into one variable.
Denoting the set in $\gamma$ that contains $z$ by $C_z$, we then have
\[
   \log \frac{\nu(w,x,y ; C_z)}{\mu_\ast(C_z)} \; = \; 
   \log \frac{\nu(w,x,y ; C_z)}{\hat{\nu}_{W,X}(w, x ; C_z)} + \log \frac{\hat{\nu}_{W,X}(w, x ; C_z)}{\hat{\nu}_W(w ; C_z)} + 
   \log \frac{\hat{\nu}_W(w ; C_z)}{ \mu_\ast(C_z)} .
\]
By integration we obtain 
\begin{eqnarray}
\lefteqn{\int_{\mathsf{Z}}  \log \frac{\nu(w,x,y ; C_z)}{\mu_\ast(C_z)} \, \nu(w, x,y; {\rm d} z) }  \nonumber \\
       & = &  \int_{\mathsf{Z}} \log \frac{\nu(w,x,y ; C_z)}{\hat{\nu}_{W,X}(w, x ; C_z)} \, \nu(w,x,y; {\rm d} z)  + \nonumber \\
       &    &    \qquad \qquad   \qquad \qquad    \int_{\mathsf{Z}} \log \frac{\hat{\nu}_{W,X}(w, x ; C_z)}{\hat{\nu}_W(w ; C_z)} \, \nu(w,x,y; {\rm d} z) + \nonumber \\
       &    &   \qquad \qquad \qquad \qquad  \qquad \qquad \qquad \qquad  \int_{\mathsf{Z}} \log \frac{\hat{\nu}_W(w ; C_z)}{\mu_\ast(C_z)} \, \nu(w,x,y; {\rm d} z)  \nonumber \\
       & = &  \sum_{C \in \gamma} \nu(w,x,y ; C)  \log \frac{ \nu(w,x,y ; C)}{\hat{\nu}_{W,X}(w,x ; C)}  + \label{firstterm} \\
       &    &     \qquad \qquad  \qquad \qquad   \sum_{C \in \gamma} \nu(w, x, y ; C)  \log \frac{\hat{\nu}_{W,X}(w, x ; C)}{\hat{\nu}_W(w; C)} + \label{secondterm}\\ 
       &    &     \qquad \qquad  \qquad \qquad \qquad \qquad  \qquad \qquad \sum_{C \in \gamma} \nu(w, x, y ; C)  \log \frac{\hat{\nu}_W(w; C)}{\mu_\ast(C)} \label{thirdterm} 
\end{eqnarray}
An integration of the last term (\ref{thirdterm}) with respect to $\mu$
yields, by the same reasoning as in the steps (\ref{step1}), (\ref{gener}), and (\ref{step2}), 
\[
     I_\gamma(W \to Z) \; = \; \sum_{C \in \gamma} \int_{\mathsf{W}} \hat{\nu}_W(w ; C)  \log \frac{\hat{\nu}_W(w; C)}{\mu_\ast(C)} \, \mu_W({\rm d} w).
\]
A corresponding integrating of the first term (\ref{firstterm}) with respect to $\mu$ yields a non-negative quantity that can be interpreted as    
$I_\gamma(Y \to Z | W, X)$ (see definition (\ref{condmutinfxyz2})). Even though we will have to slightly adjust this first term, 
the problem we are facing here is most clearly highlighted by the second term, (\ref{secondterm}).  
In order to naturally generalise the chain rule (\ref{chainnew}) we have to interpret the integral of the
second term 
as $I_\gamma(X \to Z | W)$. However, it turns out that, in general, 
\begin{eqnarray}
   I_\gamma(X \to Z | W) 
      &    =    &  \sum_{C \in \gamma} \int_{\mathsf{W} \times \mathsf{X}} \hat{\nu}_{W,X}(w, x ; C)  
                 \log \frac{\hat{\nu}_{W,X}(w, x; C)}{\hat{\nu}_W(w; C)} \, \mu_{W,X}({\rm d} w, {\rm d} x) \\
      & \not= &   \sum_{C \in \gamma} \int_{\mathsf{W} \times \mathsf{X} \times \mathsf{Y}} \nu(w, x, y ; C)  
                 \log \frac{\hat{\nu}_{W,X}(w, x; C)}{\hat{\nu}_W(w; C)} \, \mu_{W,X}({\rm d} w, {\rm d} x, {\rm d} y), \label{secondterm2}         
\end{eqnarray}
where (\ref{secondterm2}) is the integral of the term (\ref{secondterm}) with respect to $\mu$. 
We cannot even ensure that this integral is non-negative. The reason is that 
the $\sigma$-algebra used for the definition of $\hat{\nu}_W(w; C)$ is not necessarily a $\sigma$-subalgebra of the one used for the definition of the kernel 
$\hat{\nu}_{W,X}(w, x ; C)$ (the situation is similar to the one of Example \ref{notsubalgebra}). Therefore, the reasoning of the steps (\ref{step1}), (\ref{gener}), and (\ref{step2}), cannot be applied here.   
\medskip

The problem highlighted in this section will now be resolved. 
This will be done by a modification of the involved $\sigma$-algebras, which should define a filtration in order to imply a general causal version of the chain rule. 
In the next section, this modification will be presented for the general case of $n$ input variables. 

\section{The general definition of causal information flow}
\subsection{Filtrations and information} \label{sigmaalg}
Let $(\mathsf{X}_i, \rscr{X}_i)$, $i \in N := [n] = \{1,\dots,n\}$, be a family of measurable spaces, the state spaces of the input variables. 
For each subset $M$ of $N$, we have the 
corresponding product space $(\mathsf{X}_M, \rscr{X}_M)$ consisting of 
$\mathsf{X}_M := \times_{i \in M} \mathsf{X}_i$ and $\rscr{X}_M := \otimes_{i \in M} \rscr{X}_i$. 
Note that for $M = \emptyset$, the set $\mathsf{X}_\emptyset$ consists of one element, the empty sequence $\epsilon$, and 
$\rscr{X}_\emptyset = \{\emptyset, \{\epsilon\} \}$ is the trivial $\sigma$-algebra with two elements. In addition to the input variables, we consider an output 
variable with state space $(\mathsf{Z}, \rscr{Z})$. The input-output channel is given by a Markov kernel 
\[
     \nu: \; \mathsf{X}_N \times \rscr{Z} \; \to \; [0,1].
\]
Together with a probability measure $\mu$ on $(\mathsf{X}_N, \rscr{X}_N)$ this 
defines the probability space $(\Omega, \rscr{F}, {\Bbb P})$ where  
\[
     \Omega \; := \; \mathsf{X}_N \times \mathsf{Z}, \qquad \rscr{F} \; := \; \rscr{X}_N \otimes \rscr{Z}, 
\]
and  
\[ 
     {\Bbb P}(A \times C) \; := \; \int_A \nu(x;C) \, \mu({\rm d} x), \qquad A \in \rscr{X}_N, \;\; C \in \rscr{Z}. 
\]
Finally, we have the canonical projections 
\[
     X_M : \Omega \; \to \; \mathsf{X}_M, \quad M \subseteq N, \qquad \mbox{and} \qquad Z: \Omega \; \to \; \mathsf{Z}.
\]
\medskip

We are now going to define the $M$-marginal of the channel based on a general $\sigma$-subalgebra $\rscr{B}_M$ of $\rscr{X}_M$. 
Below, in Section \ref{coupling}, this will allow us to incorporate causal aspects of $\nu$ by an appropriate adaptation of $\rscr{B}_M$ to $\nu$.  
In order to highlight the flexibility that we have here, let us begin with the usual definition where $\rscr{B}_M$ equals the largest  
$\sigma$-subalgebra of $\rscr{X}_M$, that is $\rscr{X}_M$ itself. 
Given a measurable set $C \subseteq \mathsf{Z}$, we have the conditional expectation
\begin{equation} \label{bederwa}
    {\Bbb E}(\mathbbm{1}_{\{Z \in C\}} | X_{M}). 
\end{equation}
This is by definition an $X_M$-measurable function $\Omega \to {\Bbb R}$. By the factorisation lemma we can represent it 
as a composition $ {\Bbb E}(\mathbbm{1}_{\{Z \in C\}} | X_{M}) =  
\nu_M(X_M ; C)$
with a measurable function $\nu_M( \cdot ; C) : (\mathsf{X}_M, \rscr{X}_M) \to {\Bbb R}$. This allows us to define the conditional distribution
\begin{equation} \label{conddistri}
     {\Bbb P}(Z \in C | X_M = x_M) \; := \; \nu_M( x_M ; C)    
\end{equation}
which can be interpreted as a channel 
\[
   \nu_M: \; \mathsf{X}_M \times \rscr{Z} \; \to \; [0,1], \qquad (x_M,C) \; \mapsto \; \nu_M(x_M ;C).
\]
We now modify the outlined marginalisation of $\nu$ by reducing the maximal $\sigma$-algebra 
$\rscr{X}_M$ to the $\sigma$-subalgebra $\rscr{B}_M$. More precisely, we replace $X_M$ in (\ref{bederwa}) by 
\begin{equation} \label{projection}
      \widehat{X}_M: \; (\Omega, \rscr{F}, {\Bbb P}) \; \to \; (\mathsf{X}_M, \rscr{B}_M) 
\end{equation}  
and consider the conditional expectation
\begin{equation} \label{bederwa2}
    {\Bbb E}(\mathbbm{1}_{\{Z \in C\}} | \widehat{X}_{M}). 
\end{equation}
This is now an $\widehat{X}_M$-measurable function $\Omega \to {\Bbb R}$, and, by the factorisation lemma, we can represent it 
as a composition ${\Bbb E}(\mathbbm{1}_{\{Z \in C\}} | \widehat{X}_{M}) =  
\hat{\nu}( \widehat{X}_M ; C)$
with a measurable function $\hat{\nu}( \cdot ; C) : (\mathsf{X}_M, \rscr{B}_M) \to {\Bbb R}$. Finally, we have the modification 
\[
     {\Bbb P}(Z \in C | \widehat{X}_M = x_M) \; := \;  \hat{\nu}( x_M ; C)   
\] 
of the conditional distribution (\ref{conddistri}), which corresponds to a modified channel
\[
   \hat{\nu}_M: \; \mathsf{X}_M \times \rscr{Z} \; \to \; [0,1], \qquad (x_M ,C) \; \mapsto \; \hat{\nu}_M(x_M ;C).  
\]
By construction, $\hat{\nu}_M$ is $\rscr{B}_M$-measurable, which means that it uses only information that is contained in $\rscr{B}_M$. 
For the maximal $\sigma$-algebra we recover $\nu_M$. We can also consider the other extreme where $\rscr{B}_M$ equals 
the smallest $\sigma$-algebra, $\{\emptyset , \mathsf{X}_M\}$. In that case, we obtain $\hat{\nu}(x_M ; C) = \mu_\ast(C)$. 
An adjustment of $\rscr{B}_M$ to the information actually used by $\nu$ will allow us to interpret 
$\hat{\nu}$ causally. In contrast, if we do not have such an adjustment, $\hat{\nu}_M$ will represent a hypothetical channel, a ``ghost channel'', based on the
$\sigma$-algebra of an external observer rather than the $\sigma$-algebra of the actual mechanisms of the channel. \\

We now consider a family $\rscr{B} = {(\rscr{B}_M)}_{M \subseteq N}$ of $\sigma$-algebras. It 
gives rise to a corresponding family 
\[
     \rscr{F}_M \; := \; X_M^{-1}(\rscr{B}_M) \; \subseteq \; \rscr{F}, \qquad M \subseteq N,
\]
of $\sigma$-algebras on $\Omega$. We call the family $\rscr{B}$ {\em projective\/}, if 
the maps 
\[
     \pi^M_L : \mathsf{X}_M \to \mathsf{X}_L, \qquad x_M = {(x_i)}_{i \in M} \mapsto x_L = {(x_i)}_{i \in L}, \qquad L \subseteq M \subseteq N,
\]
are $\rscr{B}_M$-$\rscr{B}_L$-measurable. For projective families, we have the following monotonicity: 
\begin{equation} \label{monot}
     L \subseteq M \qquad \Rightarrow \qquad \rscr{F}_L \subseteq \rscr{F}_M.  
\end{equation}
Given a projective family $\rscr{B}$, we now define a corresponding family of information-theoretic quantities which generalise 
(conditional) mutual information. 
We begin with a local version, applied to a measurable 
partition $\gamma$ of $\mathsf{Z}$. For $z \in \mathsf{Z}$, we denote the set in $\gamma$ that contains $z$ by $C_z$. 
For $L \subseteq M \subseteq N$, we consider $x_M = (x_L, x_{M \setminus L}) \in \mathsf{X}_M$ and define 
\[
   i_\gamma(x_{M \setminus L} \to z | x_L) \; := \;  \log \frac{\hat{\nu}_M(x_M; C_z)}{\hat{\nu}_L(x_L; C_z)} .
\]
This is a local version of the conditional mutual information. Integration over $z$ yields 
\[
 \int_{\mathsf{Z}}  \log \frac{\hat{\nu}_M(x_M ; C_z)}{\hat{\nu}_L(x_L; C_z)} \, \nu(x ; {\rm d} z) 
    \; = \; \sum_{C \in \gamma} \nu(x ; C)  \log \frac{\hat{\nu}_M(x_M ; C)}{\hat{\nu}_L(x_L; C)}. 
\]
With a second integration, with respect to $\mu$, we obtain 
\begin{eqnarray*}
   \lefteqn{\sum_{C \in \gamma} \int_{\mathsf{X}} \nu(x ; C)  \log \frac{\hat{\nu}_M(x_M ; C)}{\hat{\nu}_L(x_L; C)} \, \mu({\rm d} x)} \\
    & = &  \sum_{C \in \gamma} \int_{\Omega} {\Bbb E}( \mathbbm{1}_{\{Z \in C\}} | X)   
               \log \frac{{\Bbb E}(\mathbbm{1}_{\{Z \in C\}} | \widehat{X}_M)}{{\Bbb E}(\mathbbm{1}_{\{Z \in C\}} | \widehat{X}_L )} \, {\rm d}{\Bbb P}  \\
    & = &  \sum_{C \in \gamma} \int_{\Omega} {\Bbb E}( \mathbbm{1}_{\{Z \in C\}} | \widehat{X}_M)   
               \log \frac{{\Bbb E}(\mathbbm{1}_{\{Z \in C\}} | \widehat{X}_M)}{{\Bbb E}(\mathbbm{1}_{\{Z \in C\}} | \widehat{X}_L )} \, {\rm d}{\Bbb P}  \qquad 
                \mbox{(by the monotonicity (\ref{monot}))} \\
    & = &  \sum_{C \in \gamma} \int_{\mathsf{X}_M} \hat{\nu}_M(x_M ; C)  \log \frac{\hat{\nu}_M(x_M ; C)}{\hat{\nu}_L(x_L; C)} \, \mu_M({\rm d} x_M) ,            
\end{eqnarray*}
where $\mu_M$ denotes the $M$-marginal of $\mu$. 
This suggests the following version of the conditional mutual information which we refer to as information flow.    

\begin{definition} \label{defin}
Let $\gamma$ be a finite measurable partition of $\mathsf{Z}$, and let 
$L \subseteq M \subseteq N$. Then we define the {\em information flow from $X_{M \setminus L}$ to $Z$ given $X_L$\/} as
\begin{equation}
     I_\gamma(X_{M\setminus L} \to Z | X_L) \; := \; 
     \sum_{C \in \gamma} \int_{\mathsf{X}_M} \hat{\nu}_M(x_M ; C)  \log \frac{\hat{\nu}_M(x_M ; C)}{\hat{\nu}_L(x_L; C)} \, \mu_M({\rm d} x_M) . 
\end{equation}
For $L = \emptyset$ we simplify the notation by $I_\gamma(X_{M\setminus L} \to Z)$ and refer to the {\em information flow from $X_M$ to $Z$\/}. 
\end{definition} 

Given disjoint subsets $M_1, M_2, \dots , M_k$ of $N$, we use a filtration of $\sigma$-algebras 
for proving a general chain rule for information flows. 

\begin{theorem}[General chain rule] \label{genchainrule}
Consider a projective family $\rscr{B}$ and let 
$M_1, M_2, \dots , M_k$ be disjoint subsets of $N$. Then 
\begin{eqnarray}
    \lefteqn{I_\gamma(X_{M_1}, \dots, X_{M_k}  \to Z) } \nonumber \\ 
    & = & I_\gamma(X_{M_1} \to Z) + I_\gamma(X_{M_2} \to Z | X_{M_1})  + \cdots + I_\gamma(X_{M_k} \to Z | X_{M_1}, X_{M_2}, \dots, X_{M_{k - 1}}).    \label{chainr}
\end{eqnarray}
\end{theorem}
\begin{proof} Let $M^j := \cup_{i = 1}^j M_i$, $j = 0,1,\dots,k$. The monotonicity (\ref{monot}) implies that the sequence 
\[
      \rscr{F}_j := \rscr{F}_{M^j} := X_{M^j}^{-1}(\rscr{B}_M), \qquad j = 0,\dots, k,
\] 
is increasing and therefore represents a filtration of $\sigma$-algebras. This implies
\begin{eqnarray*}
\lefteqn{I_\gamma(X_{M_1},\dots,X_{M_k} \to Z)} \\ 
   & = & \sum_{C \in \gamma} \int_{\mathsf{X}_{M^k}}  \hat{\nu}(x_{M^k} ; C) \log \frac{\hat{\nu}(x_{M^k} ; C)}{\mu_\ast (C)}  \, \mu ({\rm d} x_{M^k}) \\
   & = & \sum_{C \in \gamma} \int_{\mathsf{X}_{M^k}}  \hat{\nu}(x_{M^k} ; C) \log 
             \left( \prod_{j = 1}^k \frac{\hat{\nu}_{M^j}(x_{M^j} ; C)}{\hat{\nu}_{M^{j-1}}(x_{M^{j-1}} ; C)} \right)  \, \mu ({\rm d} x_{M^k}) \\
   & = & \sum_{j = 1}^k \sum_{C \in \gamma} \int_{\mathsf{X}_{M^k}}  \hat{\nu}(x_{M^k} ; C) 
             \log \frac{\hat{\nu}_{M^j}(x_{M^j} ; C)}{\hat{\nu}_{M^{j-1}}(x_{M^{j-1}} ; C)} \, \mu ({\rm d} x_{M^k}) \\
      & = & \sum_{j = 1}^k \sum_{C \in \gamma} \int_{\Omega}  {\Bbb P}(Z \in C | \widehat{X}_{M^k}) 
              \log \frac{{\Bbb P}(Z \in C | \widehat{X}_{M^j})}{{\Bbb P}(Z \in C | \widehat{X}_{M^{j-1}})} \, {\rm d} {\Bbb P} \\
   & = & \sum_{j = 1}^k \sum_{C \in \gamma} \int_{\Omega}  {\Bbb P}(Z \in C | \widehat{X}_{M^j}) 
              \log \frac{{\Bbb P}(Z \in C | \widehat{X}_{M^j})}{{\Bbb P}(Z \in C | \widehat{X}_{M^{j - 1}})} \, 
                {\rm d} {\Bbb P} \\
   & = & \sum_{j = 1}^k \sum_{C \in \gamma} \int_{\mathsf{X}_{M^j}}  \hat{\nu}_{M^j}(x_{M^j} ; C) 
             \log \frac{\hat{\nu}_{M^j}(x_{M^j} ; C)}{\hat{\nu}_{M^{j-1}}(x_{M^{j-1}} ; C)} \, \mu ({\rm d} x_{M^j}) \\               
   & = &  \sum_{j = 1}^k I_\gamma(X_{M_j} \to Z | X_{M_1},\dots, X_{M_{j - 1}}) .              
\end{eqnarray*}
\end{proof}

We now state basic properties of the information flow. Some of these properties are listed in \cite{JBGWS13}  as natural postulates (P0--P4) for a measure of 
causal strength.

\begin{proposition}[Natural properties] \label{natprop}
The following properties hold:
\begin{enumerate}  
\item[(a)] The information flow from all input variables to the output variable coincides with the mutual information: 
       $I_\gamma(X_{N} \to Z) = I_\gamma (X_N ; Z)$. 
\item[(b)] For a subset $M$ of $N$, the set of all input variables, 
        the information flow $I_\gamma(X_{M} \to Z )$ is smaller than or equal to the mutual information $I_\gamma(X_{M} ; Z)$.       
\item[(c)] For a subset $M$ of $N$,  
        the information flow $I_\gamma(X_{M} \to Z | X_{N \setminus M})$ is greater than or equal to the conditional mutual information 
        $I_\gamma(X_{M} ; Z | X_{N \setminus M})$.
\item[(d)] If the information flow $I_\gamma(X_{M} \to Z | X_{N \setminus M})$ vanishes then $Z$ is independent of $X_M$ given $X_{N \setminus M}$.         
\item[(e)] Let $L \subseteq M \subseteq N$. If $I_\gamma(X_{M} \to Z ) = 0$ then $I_\gamma(X_{L} \to Z ) = 0$.             
\end{enumerate}   
\end{proposition}
\begin{proof} 
Statement (a) follows from $\hat{\nu}_N(x_N ; C) = \nu (x_N ; C)$ and $\hat{\nu}_\emptyset(x ; C) = \mu_\ast (C)$. The statements (b) and (c) 
can be proven in the same way as the corresponding inequalities (\ref{compa}) of Proposition \ref{firstchain}, thereby using the chain rule  
\begin{equation*} \label{cha}
     I_\gamma(X_N \to Z) \; = \; I_\gamma (X_M \to Z) + I_\gamma(X_{N \setminus M} \to Z | X_M)
\end{equation*}
for $M \subseteq N$ (this follows from the general chain rule (\ref{chainr}),
with $M_1 = M$ and $M_2 = N \setminus M$).
In order to prove (d), note that with (c) we have 
\[
     I_\gamma(X_{M} \to Z | X_{N \setminus M}) \, = \, 0 \quad \Rightarrow \quad I_\gamma(X_{M} ; Z | X_{N \setminus M}) \, = \, 0. 
\] 
This implies that $X_M$ is independent of $Z$ given $X_{N \setminus M}$. Finally, (e) follows from the chain rule 
\begin{equation*} \label{cha}
     I_\gamma(X_M \to Z) \; = \; I_\gamma (X_L \to Z) + I_\gamma(X_{M \setminus L} \to Z | X_L),
\end{equation*} 
by the general chain rule (\ref{chainr}), with $M_1 = L$ and $M_2 = M \setminus L$.
\end{proof}

\subsection{Adaptation of the filtration to the channel} \label{coupling}
We are now going to couple the family ${(\rscr{B}_M)}_{M \subseteq N}$ to the channel $\nu$ so that we can interpret the corresponding 
marginals ${(\hat{\nu}_M)}_{M \subseteq N}$ causally.  
In order to simplify the presentation, we first consider an arbitrary 
$\sigma$-subalgebra $\rscr{A}$ of $\rscr{X}_N$. (Below, $\rscr{A}$ will be chosen to be the $\sigma$-algebra generated by $\nu$.) 
We begin with information in $M$ in the context of a configuration $\bar{x}$ outside of $M$, that is $\bar{x} \in \mathsf{X}_{N \setminus M}$. 
Given such an $\bar{x}$, we define 
the {\em $(M,\bar{x})$-trace\/} of $\rscr{A}$ as follows: 
For each $A \in \rscr{A}$, we consider the {\em $(M,\bar{x})$-section\/} of $A$, 
\[
      {\rm sec}_{M, \bar{x}}(A) \; := \; A_{M, \bar{x}} \; := \; \left\{ x \in \mathsf{X}_M \; : \; (x, \bar{x}) \in A \right\}.
\]  
These sections then form the $(M,\bar{x})$-trace of $\rscr{A}$, that is
\[
    {\rm tr}_{M, \bar{x}}(\rscr{A}) \; := \; \rscr{A}_{M, \bar{x}} \; := \; \left\{  A_{M, \bar{x}} \; : \;  A \in \rscr{A} \right\}.
\]
Considering all possible contexts $\bar{x} \in \mathsf{X}_{N \setminus M}$, we finally define the {\em $M$-trace\/} of $\rscr{A}$ as
\[
    {\rm tr}_{M}(\rscr{A}) \; := \; \rscr{A}_M \; := \; \bigvee_{\bar{x} \in \mathsf{X}_{N \setminus M}}  \rscr{A}_{M, \bar{x}}. 
\]
The $(M,\bar{x})$-trace as well as the $M$-trace of $\rscr{A}$ are $\sigma$-subalgebras of $\rscr{X}_M$. Note that in the extreme cases $M = \emptyset$ and $M = N$, 
we recover $\rscr{A}_\emptyset = \{\emptyset, \{ \epsilon \} = \mathsf{X}_\emptyset\}$ (where $\epsilon$ denotes the empty sequence), and $\rscr{A}_N = \rscr{A}$, respectively.  
\medskip

The family of all $M$-traces of $\rscr{A}$ describes how $\rscr{A}$ is ``distributed'' over the subsets $M$ of $N$. However, there is a problem here:     
The canonical projections $\pi^M_L$
are not necessarily $\rscr{A}_M$-$\rscr{A}_L$-measurable. This projectivity 
property is required for the definition of a measure of causal information flow that satisfies the general 
chain rule of Theorem \ref{genchainrule}. We highlighted this problem for the three-input case in Section \ref{threeinp}. 
There are two ways to recover the projectivity, first by extending and second by reducing  
$\rscr{A}_M$ appropriately. Let us begin with the extension: 
\begin{equation} \label{min}
         \overline{\rscr{A}}_M \; := \; \bigvee_{L \subseteq M}  \left( {\pi^M_L} \right)^{-1} \left( \rscr{A}_L \right).
\end{equation}
We have the following characterisation of the family $\overline{\rscr{A}}_M$, $M \subseteq N$, as the smallest 
projective extension of the family $\rscr{A}_M$, $M \subseteq N$. 

\begin{proposition}[Extension of $\rscr{A}_M$, $M \subseteq N$] 
The family $\overline{\rscr{A}}_M$, $M \subseteq N$, satisfies the following two conditions:
\begin{enumerate}
\item For all $M \subseteq N$, $\rscr{A}_M$ is contained in $\overline{\rscr{A}}_M$.  
\item For all $L \subseteq M \subseteq N$, the canonical projection $\pi^M_L$ is $\overline{\rscr{A}}_M$-$\overline{\rscr{A}}_L$-measurable.
\end{enumerate}
Furthermore, for every family $\rscr{A}_M'$, $M \subseteq N$, that satisfies these two conditions (where $\overline{\rscr{A}}_M$ is replaced by $\rscr{A}'_M$), we have 
\begin{equation} \label{minim}
      \overline{\rscr{A}}_M \; \subseteq \; \rscr{A}_M', \qquad \mbox{for all $M \subseteq N$}.  
\end{equation}
\end{proposition}
\begin{proof} The first statement is clear (simply choose on the RHS of (\ref{min}) $L = M$). For the second statement, we have to show
\[
     \left( \pi^M_L \right)^{-1} \left( \overline{\rscr{A}}_L \right) \; = \; \overline{\rscr{A}}_M . 
\]
Given that 
\[
      \overline{\rscr{A}}_L  \; = \; \bigvee_{K \subseteq L}  \left( {\pi^L_K} \right)^{-1} \left( \rscr{A}_K \right)
\]
it is sufficient to verify
\begin{equation} \label{inclus1}
     \left( \pi^M_L \right)^{-1} \left(   \left( {\pi^L_K} \right)^{-1} \left( \rscr{A}_K \right) \right) \; \subseteq \; \overline{\rscr{A}}_M \qquad \mbox{for all $K \subseteq L$}. 
\end{equation}
The LHS of (\ref{inclus1}) reduces to $\left( \pi^M_K \right)^{-1} \left( \rscr{A}_K \right)$ which is by definition contained in $\overline{\rscr{A}}_M$. \\
Finally, we prove the minimality. It is easy to see that any family $\rscr{A}'_M$, $M \subseteq N$, that satisfies the two conditions has to contain 
$\left( \pi^M_L \right)^{-1} \left( \rscr{A}_L \right)$, $L \subseteq M$. By definition, $\overline{\rscr{A}}_M$ is the smallest $\sigma$-algebra that contains these 
$\sigma$-subalgebras (see (\ref{min})). This implies (\ref{minim}). 
\end{proof}

After having defined the smallest extension of the family $\rscr{A}_M$, $M \subseteq N$, as one way to recover projectivity, 
we now come to the alternative way, which is by reduction of that family. More precisely, we define 
\begin{equation} \label{max}
    \underline{\rscr{A}}_M \; := \; \left\{ A \in \rscr{X}_M \; : \; \left(\pi^{N}_M\right)^{-1}(A) \in \rscr{A} \right\} .
\end{equation}   
We have the following characterisation of this family as the largest projective reduction of $\rscr{A}_M$, $M \subseteq N$.  

\begin{proposition}[Reduction of $\rscr{A}_M$, $M \subseteq N$] 
The family $\underline{\rscr{A}}_M$, $M \subseteq N$, satisfies the following two conditions:
\begin{enumerate}
\item For all $M \subseteq N$, $\underline{\rscr{A}}_M$ is contained in $\rscr{A}_M$.  
\item For all $L \subseteq M \subseteq N$, the canonical projection $\pi^M_L$ is $\underline{\rscr{A}}_M$-$\underline{\rscr{A}}_L$-measurable.
\end{enumerate}
Furthermore, for every family $\rscr{A}_M'$, $M \subseteq N$, that satisfies these two conditions (where $\underline{\rscr{A}}_M$ is replaced by $\rscr{A}'_M$), we have 
\begin{equation}
     \rscr{A}_M' \; \subseteq \;  \underline{\rscr{A}}_M , \qquad \mbox{for all $M \subseteq N$}.  
\end{equation}
\end{proposition}
\begin{proof} In order to prove the first statement, let $A \in \underline{\rscr{A}}_M$. 
This means that 
\[
    \widetilde{A} \; := \; \left(\pi^{N}_M\right)^{-1}(A) \; = \;  A \times \mathsf{X}_{N \setminus M} \; \in \; \rscr{A}. 
\]
For all $\bar{x} \in \mathsf{X}_{N \setminus M}$, we have
\[
   {\rm sec}_{M, \bar{x}} (\widetilde{A}) \; = \; \left\{ x \in \mathsf{X}_M \; : \; (x, \bar{x}) \in \widetilde{A} \right\} \; = \; A .
\]
This means that $A \in {\rm tr}_M (\rscr{A}) = \rscr{A}_M$, which concludes the proof of the first statement. Now we come to the measurability of the 
canonical projection $\pi^M_L$. For this, we choose $A \in \underline{\rscr{A}}_L$ and have to show $\left(\pi^M_L\right)^{-1}(A) \in \underline{\rscr{A}}_M$:
\[
    \left(\pi^N_M\right)^{-1} \left( \left(\pi^M_L\right)^{-1}(A) \right) \; = \; \left( \pi^N_L \right)^{-1}(A) \; \in \; \rscr{A} \qquad \mbox{(by definition (\ref{max}))}.  
\]
Finally, we have to prove the maximality. Let $\rscr{A}_M'$, $M \subseteq N$, be a family that satisfies the two conditions. Then 
\[
    \left( \pi^N_M \right)^{-1} \left( \rscr{A}'_M\right) \; \subseteq \; \rscr{A}_N' \; \subseteq \; \rscr{A}_N \; = \; \rscr{A}. 
\] 
This means that $\rscr{A}'_M \subseteq \underline{\rscr{A}}_M$.
\end{proof}
\medskip

This concludes the constructions for a given $\sigma$-algebra $\rscr{A}$, without explicit reference to the channel
$\nu: \mathsf{X} \times \rscr{Z} \to [0,1]$. 
We now couple the studied $\sigma$-algebras  
with the channel $\nu$ and therefore choose $\rscr{A}$ to be  
the $\sigma$-algebra generated by the channel $\nu$, that is $\sigma(\nu)$. 
We highlight this coupling by writing $\rscr{A}^\nu$, as a particular choice of $\rscr{A}$, and consider the 
family ${(\rscr{A}_M^\nu)}_{M \subseteq N}$ of its traces, together with the corresponding smallest projective extension 
${(\overline{\rscr{A}}_M^\nu)}_{M \subseteq N}$ and the largest projective reduction ${(\underline{\rscr{A}}_M^\nu)}_{M \subseteq N}$. 
In the context of a channel, the traces of $\rscr{A}^\nu$ have a natural interpretation. In order to see this, we first consider a configuration 
$\bar{x} \in \mathsf{X}_{N \setminus M}$ and define the ``constrained'' Markov kernel   
\[
    \nu_{M, \bar{x}}: \; \mathsf{X}_M \times \rscr{Z}  \; \to \; [0,1], \qquad  \nu_{M, \bar{x}}(x ; C) \; := \; \nu(x, \bar{x} ; C). 
\] 
We denote the $\sigma$-algebra generated by $\nu_{M, \bar{x}}$ by $\sigma_{M, \bar{x}}(\nu)$. Taking all ``constraints'' $\bar{x}$ into account, 
we then define 
\[
      \sigma_M(\nu) \; := \; \bigvee_{\bar{x} \in \mathsf{X}_{N \setminus M}} \sigma_{M, \bar{x}}(\nu). 
\]    

\begin{proposition} Let $\rscr{A}^\nu \subseteq \rscr{X}_N$ be the $\sigma$-algebra generated by the Markov kernel 
$\nu: \mathsf{X}_N \times \rscr{Z} \to [0,1]$. Then for all $M \subseteq N$ and all $\bar{x} \in \mathsf{X}_{N \setminus M}$,  
\[
     \sigma_{M, \bar{x}}(\nu) \, = \, {\rm tr}_{M, \bar{x}} (\rscr{A}^\nu) \quad \mbox{and} \quad \sigma_{M}(\nu) \, = \, {\rm tr}_{M} (\rscr{A}^\nu).   
\]
\end{proposition}
\begin{proof} The $\sigma$-algebra $\rscr{A}^\nu$ is the smallest $\sigma$-algebra that contains all measurable sets of the form
\begin{equation} \label{form}
       A \; = \; \left\{ x \in \mathsf{X}_N \; : \; \nu(x ; C) \in B \right\},
\end{equation} 
with some $C \in \rscr{Z}$ and a Borel set $B$ in $\rscr{B}([0,1])$. Now consider the $(M, \bar{x})$-section of such a set $A$:
\begin{eqnarray*}
    {\rm sec}_{M, \bar{x}}(A)  
    & = & \{ x \in \mathsf{X}_{N \setminus M} \; : \; (x, \bar{x}) \in A \} \\
    & = &  \{ x \in \mathsf{X}_{N \setminus M} \; : \; \nu(x, \bar{x} ; C) \in B \} \\
    & = &  \{ x \in \mathsf{X}_{N \setminus M} \; : \; \nu_{M, \bar{x}}(x ; C) \in B \}.
\end{eqnarray*} 
This shows that the sections ${\rm sec}_{M, \bar{x}}(A)$ of measurable sets $A$ of the form (\ref{form}) generate $\sigma_{M, \bar{x}}(\nu)$, which proves the first equality. The second equality is a direct implication of the first one. 
\end{proof}

The results of the previous section, Theorem \ref{genchainrule} and Proposition \ref{natprop}, apply to the information flows, 
defined for the projective families ${(\overline{\rscr{A}}_M^\nu)}_{M \subseteq N}$ and ${(\underline{\rscr{A}}_M^\nu)}_{M \subseteq N}$. 
These families take into account the information that is actually used by the channel 
$\nu$. Therefore, we can interpret the corresponding marginal channels $\hat{\nu}_M$ causally, 
where we have to distinguish two kinds of causality. For the projective family ${(\overline{\rscr{A}}_M^\nu)}_{M \subseteq N}$, the channel $\hat{\nu}_M$ incorporates the information in any input configuration $x_K$, $K \subseteq M$, 
that is used by $\nu$ in conjunction with a context configuration $\bar{x}_{N \setminus K} = \bar{x}_{N \setminus M} \bar{x}_{M \setminus K}$ outside of $K$. 
For the projective family 
${(\underline{\rscr{A}}_M^\nu)}_{M \subseteq N}$, on the other hand, the channel $\hat{\nu}_M$ incorporates the information used by $\nu$ that is solely contained in $x_M$, independent of any context. When comparing a marginal channel $\hat{\nu}_M$ with another marginal channel $\hat{\nu}_L$, where $L \subseteq M$, the corresponding information flows       
$\overline{I}_\gamma (X_{M\setminus L} \to Z | X_{L})$ and $\underline{I}_\gamma (X_{M\setminus L} \to Z | X_{L})$, respectively, 
quantify the causal effects in $\hat{\nu}_M$ that exceed those in $\hat{\nu}_L$. These measures will capture different causal aspects, where the difference can be large. This is illustrated by the following extension of Example \ref{notsubalgebra}.    

\begin{example} \label{sum}
Let 
\[
     (\mathsf{X}_i, \rscr{X}_i) \; = \; ({\Bbb R}, \rscr{B}({\Bbb R})), \qquad i \in \{1,\dots, n\} = N,
\]      
where $\rscr{B}({\Bbb R})$ denotes the Borel $\sigma$-algebra of ${\Bbb R}$. We define the channel simply by the sum of the input states, 
interpreted as a Markov kernel,
\[
      \nu(x_1, \dots, x_n ; C) \; := \; \mathbbm{1}_C(x_1 + \dots + x_n). 
\]
As $\rscr{B}({\Bbb R})$ is generated by the intervals $[r - \varepsilon , r + \varepsilon] \subseteq {\Bbb R}$, the smallest $\sigma$-algebra $\rscr{A}^\nu$ for which all 
functions $\nu(\cdot ; C)$ are measurable is generated by the following sets
\[
     A({r, \varepsilon}) \; := \; \left\{ (x_1,\dots, x_n) \in {\Bbb R}^n \; : \; r - \varepsilon \, \leq \, x_1 + \dots + x_n \, \leq \, r + \varepsilon \right\}, 
     \qquad r \in {\Bbb R}, \; \varepsilon \in {\Bbb R}_+. 
\]
For a set $M \subseteq N$ and a context configuration $\bar{x} = {(\bar{x}_i)}_{i \in N \setminus M}\in {\Bbb R}^{N \setminus M}$, the $(M, \bar{x})$-section of $A_{r, \varepsilon}$ is given by 
\[
     {\rm sec}_{M, \bar{x}} (A({r, \varepsilon})) \; = \; \left\{ x = {(x_i)}_{i \in M} \in {\Bbb R}^M \; : \; 
     r - \sum_{i \in N \setminus M} \bar{x}_i - \varepsilon \, \leq \, \sum_{i \in M} x_i \, \leq \, r - \sum_{i \in N \setminus M} \bar{x}_i + \varepsilon   \right\}.
\]
Therefore, the $M$-trace of $\rscr{A}^\nu$, $\rscr{A}^\nu_M$, is generated by the halfspaces 
\[
       H_\vartheta \; := \; \left\{ x = {(x_i)}_{i \in M} \in {\Bbb R}^M \; : \;  \sum_{i \in M} x_i \leq \vartheta \right\}, \qquad \vartheta \in {\Bbb R}.
\] 
For $|M| = 1$, we recover the half lines, so that $\rscr{A}^\nu_{\{i\}} = \rscr{B}({\Bbb R})$. The projective extension then leads to the largest $\sigma$-algebra, 
the Borel algebra of ${\Bbb R}^M$: 
\[
       \overline{\rscr{A}}^\nu_M \; = \; {\rscr B}({\Bbb R}^{M}). 
\]
Therefore, the marginal channel $\hat{\nu}_M(x ; C)$ equals the usual marginal $\nu_M(x ; C)$ for the projective extension. 
For the projective reduction, on the other hand, we obtain the trivial $\sigma$-algebra except for $M = N$:
\[
      \underline{\rscr{A}}^\nu_M \; = \; 
      \left\{
         \begin{array}{c@{,\quad}l}
              \{\emptyset, {\Bbb R}^M\} & \mbox{if $M \subsetneq N$} \\
              \rscr{A}^\nu & \mbox{if $M = N$}
         \end{array}
      \right. .
\]
In this case we have 
$\hat{\nu}_M(x; C) = \mu_\ast (C)$ for $M \not= N$ and $\hat{\nu}_N(x; C) = \nu(x; C)$, 
where $\mu$ is the joint distribution of the input variables. 

We now consider the  
information flows associated with $L \subsetneq M \subseteq N$, for the projective extension as well as for the projective reduction. 
In both cases these flows coincide with usual (conditional) mutual informations, in an instructive way. 
More precisely, for the extension we have 
\begin{equation} \label{infflowext}
   \overline{I}_\gamma (X_{M\setminus L} \to Z | X_{L}) \; = \; {I}_\gamma (X_{M\setminus L} ; Z | X_{L}).   
\end{equation}
For the reduction, we obtain 
\begin{equation} \label{infflowred}
   \underline{I}_\gamma (X_{M\setminus L} \to Z | X_{L}) \; = \; 
   \left\{
   \begin{array}{c@{,\quad}l}
          0 & \mbox{if $M \subsetneq N$} \\
         {I}_\gamma (X_{N} ; Z)  & \mbox{if $M = N$}
   \end{array}
   \right. .
\end{equation}
Interestingly, (\ref{infflowred}) does not depend on $L$. The vanishing of the information flow for $M \not= N$ is due to the fact that the 
output of the channel, the sum $x_1 + \dots + x_n$, cannot be computed from a proper subset of the inputs. The flow of information only takes place
if all inputs are given. 
\end{example}

\section{Conclusions}
Conditioning is an important operation within the study of causality. The theory of causal networks, pioneered by Pearl \cite{Pe00}, introduces interventional 
conditioning as an operation, the so-called {\em do-operation\/}, 
that is fundamentally different from the classical conditioning based on the general rule ${\Bbb P}(B | A) = {\Bbb P}(A \cap B)/{\Bbb P}(A)$. 
It models more appropriately experimental setups and avoids confusion with purely associational dependencies. Information theory 
has been classically used for the quantification of such dependencies, in terms of mutual information and conditional mutual information \cite{Sh48}. 
Within the original setting of information theory, the mutual information between the input and the output of a channel can be interpreted causally. 
In the more general context of causal networks, however, confounding effects make a distinction between associations and causal effects more difficult. In such cases, 
information-theoretic quantities can be misleading as measures of causal effects. In order to overcome this problem, information theory has been coupled with the interventional calculus of causal networks, and corresponding measures of causal information flow have been proposed \cite{AK07, AP08}. Given that such 
measures are based on the notion of an experimental intervention, which represents a perturbation of the system, 
it remains unclear to what extent they quantify causal information flows in the unperturbed system. As another consequence of the interventional conditioning, one cannot expect that causal information flow, as defined in \cite{AP08}, decomposes according to a chain rule. The current article is based on an idea from 2003 
which precedes the above-mentioned works on combining the theory of causal networks with information theory. 
It proposes a way to quantify causal information flows without perturbing the system through intervention. Instead, it is based on classical conditioning
in terms of the conditional distribution ${\Bbb P}(B | \rscr{A})$, where the $\sigma$-algebra is adjusted to  
the intrinsic mechanisms of the system. The derived information flow measure satisfies the chain rule and the natural properties of a general measure of causal strength postulated in \cite{JBGWS13}. The chain rule, together with the generalised Pythagoras relation from information geometry, provide powerful 
tools within the study of the problem of partial information decomposition \cite{BROJA14, LBJW18, APV00}. 
\medskip

Even though the introduced information flows satisfy natural 
properties, the aim of the present article is relatively moderate. For instance,    
the analysis is focussed on a simple network consisting of a number of inputs and one output, which is a strong restriction compared to the setting of \cite{AP08}.
The extension of the present work to more general casual networks remains to be worked out. 
Furthermore, this article does not address the important problem of causal inference 
\cite{PJS17}. In addition to these general directions of research, there are various ways to modify and extend the constructions of the present work 
and thereby potentially highlight further causal aspects of a given channel. The following perspectives are particularly important: 
\begin{enumerate}
\item In the present article, the information flow has been defined for a fixed finite measurable partition $\gamma$ of the state space $(\mathsf{Z}, \rscr{Z})$ 
of the output variable $Z$. A natural further step would be to consider the limit of information flows
with respect to an increasing sequence $\gamma_n$, $n = 1,2,\dots$, so that 
\[
     \bigvee_{n = 1}^\infty \sigma(\gamma_n) \; = \; \rscr{Z}. 
\]  
This limit will be an information flow measure that is independent of a particular partition.  
\item Throughout this article, the partition $\gamma$ has not been coupled with the $\sigma$-algebra of the channel $\nu$. This is the smallest $\sigma$-algebra for which 
all functions $\nu(x ; C)$, $C \in \rscr{Z}$, are measurable. Given that the channel is analysed with respect to the partition $\gamma$, one can restrict attention to the smallest $\sigma$-algebra for which the functions $\nu(x ; C)$, $C \in \gamma$, are measurable. This will be a potentially small $\sigma$-subalgebra 
of the one generated by the channel. We would then have a natural 
coupling of the partition $\gamma$ with the information used by the channel.   
\item We started with the family $\rscr{A}_M^\nu$ of $M$-traces of $\rscr{A}^\nu$, the $\sigma$-algebra generated by $\nu$, 
as the natural family associated with the channel. However, 
these traces do not form a projective family of $\sigma$-algebras. Such a projectivity is required for the chain rule for corresponding information flows. 
One can recover projectivity by extension and by reduction, leading to $\overline{\rscr{A}}^\nu_M$ and $\underline{\rscr{A}}^\nu_M$, respectively. 
Example \ref{sum} shows that the extension can lead to the largest $\sigma$-algebra and the reduction to the trivial one. Given this fact, 
one might ask whether the extension is too large and the reduction is too small to capture the causal aspects of $\nu$. Even though we argued above that these two 
projective families associated with $\nu$ capture two different kinds of causal aspects, this question remains to be 
further pursued. One possible direction would be the analysis of the context-dependent traces of ${\rscr A}^\nu$, that is the family of ${\rm tr}_{M, \bar{x}}({\rscr A}^\nu)$, 
$\bar{x} \in \mathsf{X}_{N \setminus M}$. Instead of conditioning with respect to the join 
\[
     {\rm tr}_{M}({\rscr A}^\nu) \; = \; \bigvee_{\bar{x} \in \mathsf{X}_{N \setminus M}} {\rm tr}_{M, \bar{x}}({\rscr A}^\nu),
\]  
one could adjust the conditioning to the individual $\sigma$-algebras ${\rm tr}_{M, \bar{x}}({\rscr A}^\nu)$. 
This would represent an important refinement of the presented theory. 
\end{enumerate}

\bibliography{information_flow_arxiv_july_7_2020} \bibliographystyle{plainnat}

\newcommand{\etalchar}[1]{$^{#1}$}
\begin{thebibliography}{JBGWS13}

\bibitem[AA15]{AyAm15}
Nihat Ay and Shun-ichi Amari.
\newblock A novel approach to canonical divergences within information
  geometry.
\newblock {\em Entropy}, 17(12):8111?8129, Dec 2015.

\bibitem[AJVLS17]{AJLS17}
Nihat Ay, J{\"u}rgen Jost, H{\^o}ng V{\^a}n~L{\^e}, and Lorenz
  Schwachh{\"o}fer.
\newblock {\em Information Geometry}.
\newblock Springer, 2017.

\bibitem[AK07]{AK07}
Nihat Ay and David~C. Krakauer.
\newblock Geometric robustness theory and biological networks.
\newblock {\em Theory in Biosciences}, 125(2):93--121, 2007.

\bibitem[Ama16]{Am16}
Shun-ichi Amari.
\newblock {\em Information geometry and its applications}, volume 194.
\newblock Springer, 2016.

\bibitem[AN00]{AN00}
Shun-ichi Amari and Hiroshi Nagaoka.
\newblock {\em Methods of information geometry}.
\newblock Oxford University Press, 2000.

\bibitem[AP08]{AP08}
Nihat Ay and Daniel Polani.
\newblock Information flows in causal networks.
\newblock {\em Advances in complex systems}, 11(01):17--41, 2008.

\bibitem[APV20]{APV00}
Nihat Ay, Daniel Polani, and Nathaniel Virgo.
\newblock Information decomposition based on cooperative game theory.
\newblock {\em Kybernetika}, 2020.
\newblock arXiv:1910.05979.

\bibitem[BBHL16]{BBHL16}
Terry Bossomaier, Lionel Barnett, Michael Harr{\'e}, and Joseph~T Lizier.
\newblock {\em An introduction to transfer entropy}.
\newblock Springer, 2016.

\bibitem[BRO{\etalchar{+}}14]{BROJA14}
Nils Bertschinger, Johannes Rauh, Eckehard Olbrich, Jürgen Jost, and Nihat Ay.
\newblock Quantifying unique information.
\newblock {\em Entropy}, 16(4):2161?2183, Apr 2014.

\bibitem[CT06]{CT06}
Thomas~M. Cover and Joy~A. Thomas.
\newblock {\em Elements of Information Theory 2nd Edition (Wiley Series in
  Telecommunications and Signal Processing)}.
\newblock Wiley-Interscience, July 2006.

\bibitem[Dud02]{Du02}
Richard~M. Dudley.
\newblock {\em Real Analysis and Probability}.
\newblock Cambridge Studies in Advanced Mathematics. Cambridge University
  Press, 2 edition, 2002.

\bibitem[Gra69]{Gra69}
Clive W.~J. Granger.
\newblock Investigating causal relations by econometric models and
  cross-spectral methods.
\newblock {\em Econometrica}, 37(3):424--438, 1969.

\bibitem[Gra80]{Gra80}
Clive W.~J. Granger.
\newblock {Testing for causality: A personal viewpoint}.
\newblock {\em Journal of Economic Dynamics and Control}, 2(1):329--352, May
  1980.

\bibitem[JBGWS13]{JBGWS13}
Dominik Janzing, David Balduzzi, Moritz Grosse-Wentrup, and Bernhard Schölkopf.
\newblock Quantifying causal influences.
\newblock {\em Ann. Statist.}, 41(5):2324--2358, 10 2013.

\bibitem[Kak]{Ka99}
Y\^{u}ichir\^{o} Kakihara.
\newblock {\em Abstract Methods in Information Theory}, volume~4 of {\em
  Multivariate Analysis}.

\bibitem[LBJW18]{LBJW18}
Joseph Lizier, Nils Bertschinger, Jürgen Jost, and Michael Wibral.
\newblock Information decomposition of target effects from multi-source
  interactions: Perspectives on previous, current and future work.
\newblock {\em Entropy}, 20(4):307, Apr 2018.

\bibitem[Pea00]{Pe00}
Judea Pearl.
\newblock {\em Causality: Models, Reasoning, and Inference}.
\newblock Cambridge University Press, 2000.

\bibitem[PJS17]{PJS17}
J.~Peters, D.~Janzing, and B.~Sch{\"o}lkopf.
\newblock {\em Elements of Causal Inference - Foundations and Learning
  Algorithms}.
\newblock Adaptive Computation and Machine Learning Series. The MIT Press,
  Cambridge, MA, USA, 2017.

\bibitem[Sch00]{Schr00}
Thomas Schreiber.
\newblock Measuring information transfer.
\newblock {\em Phys. Rev. Lett.}, 85:461--464, 2000.

\bibitem[Sha48]{Sh48}
Claude~E. Shannon.
\newblock A mathematical theory of communication.
\newblock {\em Bell System Technical Journal}, 27(4):623--656, 1948.

\end{thebibliography}

\end{document}